\documentclass[preprint,12pt]{article}
\pdfoutput=1 
\usepackage{graphicx}
\usepackage{amsfonts}
\usepackage{amsmath}
\usepackage{url}
\usepackage{comment}
\usepackage{epstopdf}
\usepackage{hyperref} 
\usepackage{color}
\usepackage{rotating}
\usepackage{natbib}
\usepackage{footnote}
\usepackage{bm}
\setcitestyle{authoryear,open={(},close={)}} 

\def\S{{\mathbf S}}
\def\E{{\mathbf E}}

\usepackage{hyperref}
\usepackage{float}

\def\R{{\mathbb R}}

\newcommand{\clr}{\color{black}}
\newcommand{\clrr}{\color{black}}
\newcommand{\clb}{\color{black}}

\newcommand{\clf}{\color{black}} 
\newcommand{\clh}{\color{black}} 

\def\EMA{\mathrm{EMA}}
\def\SMA{\mathrm{SMA}}
\def\BB{\mathrm{BB}}
\def\ARP{\mathrm{ARP}}
\def\MACD{\mathrm{MACD}}
\def\MOM{\mathrm{MOM}}

\usepackage{amsthm}

\usepackage{graphicx}
\usepackage{subcaption}
\usepackage{listings}  
\usepackage{xcolor} 

\theoremstyle{definition}

\newtheorem{proposition}{Proposition}[section]

\newtheorem{lemma}{Lemma}[section]

\newtheorem{corollary}{Corollary}[section]

\begin{document}
\title{Breaking the Trend: How to Avoid Cherry-Picked Signals}

\author{Sebastien Valeyre \thanks{Machina Capital, 21 rue de la Banque, 75002 Paris, France, svaleyre@machinacap.com} \thanks{Valeyre Research, 114 rue d'antibes 06400  Cannes, France, sb.valeyre@gmail.com or sv@valeyreresearch.com } }



\maketitle


\date{\today}


\begin{abstract}

	Our empirical results show %
{\clf a good fit} with the {\clb theoretical} Sharpe formula of a trend-following strategy depending on the parameter of the signal, which was derived by \cite{Grebenkov}. %
{\clf That empirical fit suggests that a mean-reversion process with only one time scale is enough to model the trend-following mechanism at the average scale of CTAs and that, as a consequence, a single simple $\EMA$ appears sufficient to capture the trend. Using a complex basket of different sophisticated indicators as signal therefore does not seem to be necessary and exposes to the risk of cherry-picking.}

\end{abstract}


\thanks{I thank  Prof{\clb essor} Valeriy Zakamulin for his feedback and valuable insights. }

\newpage

\section* {Conflict of Interest disclosure statement}

I have nothing to disclose
No funding was received

\section{Introduction}
\label{intro}
{Practitioners of the trend-following strategies usually optimize their strategies empirically by selecting parameters that maximize the simulated reward-to-risk ratio,  typically measured by the Sharpe ratio, {\clb when replaying the historical scenarii through past returns}. However, this empirical approach can lead to overfitting, especially when optimization is performed on the same historical dataset used for evaluation. { Many financial researchers have proposed various methods to address this issue, which can be broadly classified into three categories, as noted by \cite{Koshiyama}: data snooping, overestimated performance, and cross-validation evaluation. They introduced a technique to mitigate financial overfitting called the Covariance-Penalty Correction, which adjusts a risk metric downward based on the number of parameters and the volume of data underpinning a trading strategy. }}

{An alternative is to model the returns of the underlying assets using stochastic processes. In some cases, this allows the Sharpe ratio to be derived mathematically through a formula depending on the parameters of the strategy. When this is possible, we refer to it as the theoretical Sharpe ratio, in contrast to the empirical one.}

{It is also possible to determine the set of parameters that maximizes the theoretical Sharpe ratio—this is known as the theoretical optimal Sharpe ratio, as opposed to the common empirical optimum obtained through backtesting.}

{A theoretical Sharpe ratio is always interesting if the model deriving the returns is realistic enough, as it allows validat{\clb ing} and optimiz{\clb ing} a strategy without the risk of overfitting.}

{The main objective of this paper is to empirically validate the theoretical Sharpe ratio model proposed by \cite{Grebenkov} and therefore their assumption to model trends as gaussian mean reversion stochastic processes with only one time scale.}

\section{Literature}

The usual recipe to determine a portfolio in the CTA's industry is to use a blend of many different technical indicators as {a} signal for each underlying instrument (most often individual indicator gives either a long or a short position\footnote{ {\clb \cite{Sepp} classified the methodologies for sizing positions based on signals into three categories:  (1) European trend-following, which assumes continuous position sizes proportional to signal strength — an approach used by many large European CTA managers; (2) Time Series Momentum (\cite{Moskowitz,Hurst13,Baltas}), a method more commonly used in academic studies, where positions depend on the sign of the momentum; and (3) American trend-following, which originates from the early adopters of CTA strategies.}}).  Then, a risk management process is applied to the signals to ensure that the portfolio is diversified enough and target{\clb s} a constant risk. One simple solution is to impose an equal conditional risk in each asset class, as the SG Trend Indicator does and to size the portfolio to target a constant volatility. Another solution is to use the correlation matrix and apply mathematical optimization to maximize the reward/risk ratio. The well-known Markowitz solution set positions as linearly depending on the signals. The linear dependence is derived simply through the normalization of the signals by the inverse of the correlation matrix between returns, but the {empirical} results are not appealing. The Agnostic Risk Parity ($\ARP$), introduced by \cite{Benichou16}, normalizes the signals through the inverse of the square root of the correlation matrix. \cite{Benichou16} proposed this portfolio because it was rotationally invariant. %
{\clf This concept comes from physics, and their portfolio is not yet widely used in the finance literature, although it is a serious alternative to the Markowitz solution.} {For example, \cite{Benveniste} explained, for the celebration of the 50th year of the Journal of Portfolio Management, the myth surrounding the alleged inadaptation of mean–variance optimization to the real world, and how efficient mean-variance optimization works without any conflict with empirical measurements, as long as the correlation matrix is properly modeled, yet neither referenced \cite{Benichou16} nor addressed the impact of uncertainty in mean within Markowitz formula. The literature has long acknowledged the sensitivity of the Markowitz formula to parameter uncertainty, without recognizing that this uncertainty can actually alter the form of the formula itself\footnote{A dedicated review of the literature is provided in the appendix.}.}

In reality, one hidden assumption of Markowitz optimization is that  {the expected means of} returns from the signals are known and certain, but this assumption is incorrect. \cite{Valeyre} proved that the Agnostic Risk Parity ($\ARP$) approach was optimal {in the context of trend-following, where expected means are not certain, are estimated with some errors and do not break rotational invariance as biases introduced by risk premia would}. This optimality holds when the correlation matrix between the trends, signals or {uncertainty of mean} is different from the correlation matrix between returns but is a kind of random matrix where a dominant factor of this correlation matrix is very dominant and diffuses {\clb randomly}. \cite{Valeyre} suggested an interacting agents model to justify such an hypothesis.


There are many different technical indicators used as proxies for expected means in trend-following, but all of them are based on past returns, which are supposed to help predict future ones, while inevitably containing noise and uncertainty. Most of the popular indicators used are based on moving averages of past prices. The most popular is the Simple Moving Average ($\SMA$), while less commonly used types include the Linear Moving Average and the Exponential Moving Average on returns ($\EMA$). Each moving average is computed using an averaging window of a particular size. Trend-following indicators can also be based on a combination of moving averages, such as crossovers (one with a short window size and another with a long window size). The switch is determined when the short-term moving average crosses the long-term moving average. As an example, The Societe Generale Trend Indicator, which is a reference among CTAs, is determined only by the crossover based on a simple average based 20 and 120 business day parameters. Momentum ($\MOM$) is also a very popular indicator and is defined as the difference between two prices. \cite{Hurst} analysed the performance simulated on the last century using a mixture of 1 {\clb month}, 3 month and 1 year momentum. \cite{Lamperiere} analysed the performance simulated on the last 2 centuries using a 5 month $\EMA$ signal\footnote{\clr They were ambiguous and only wrote ``exponential moving average~[\ldots] with a decay rate equal to $n$ months''. The parameter $n$ is then interpreted as the half-life, defined as the lookback period over which the cumulative weight of past observations reaches $50\%$. This corresponds approximately to $\frac{\ln(2)}{\eta}$, where $\eta$ is defined later in Eq. \eqref{eq:ema}.}. {\clr \cite{quantica} specifies that a half-life of 60-70 business days for an $\EMA$, corresponding to approximately one calendar quarter, best replicates the return and risk characteristics of trend-following CTA industry benchmarks such as the SG Trend Index, which invests roughly on the 10 largest systematic CTA funds with trend-following behavior, making it a robust bottom-up representation of that style}. The Bollinger Band ($\BB$) is very popular, but non continuous with only 3 outputs 0 and 1 or -1 when price is outside a band. The Moving Average Convergence/Divergence ($\MACD$) uses a combination of three $\EMA$s to capture mean reversion at both short and long-term scales while identifying trends in the medium term.

We can describe some of these indicators through their sensitivity to past daily returns as introduced by \cite{Zakamuli20}. Fig.\ref{fig:sensitivity} displays some cases. It can be observed that the Exponential Moving Average ($\EMA$) exhibits a sensitivity that declines exponentially, while the crossover resembles a hat shape. Sensitivity is low for very recent past returns, high for a moderately distant past, and negligible for a long past. The shape of the crossover can be adapted if price behavior exhibits mean-reverting tendencies on very short time scales, as seen in the case of individual stocks with low liquidity. From an academic point of view, almost all these trend-following rules are ad hoc and lack theoretical justification for their optimality. \cite{Zakamuli20} related them to each other through the analysis of their sensitivity to past daily returns  \footnote{ {\clrr \clr so a 33  or 36 business day $\EMA$ ({\clr{i.e., $\eta = \frac{1}{33}$ or $\eta = \frac{1}{36}$ in Eq. \eqref{eq:ema}}}) should be optimal for replicating the mixture of 1 month, 3 month, and 1 year momentum signals used by \cite{Hurst}, based on an equivalent average lookback or an equivalent half-life. By contrast, a 43 or 62 business-day $\EMA$ should be optimal for the SG Trend Indicator crossover}}. \cite{Zakamuli20} also reminds us that it is well known that if returns follow an autoregressive process, the best predictor has the same functional form as the autoregressive process being predicted. Therefore, the expected structure of the autoregressive process should provide a good proxy for the best predictor.

\cite{Zakamuli} shows that the indicator with the same sensitivity to past returns as the autoregressive coefficients of the return process is both the best predictor and the optimal indicator, yielding the highest Sharpe ratio when the investor has only two options (Buy or Sell with the same level of risk). He also employs a two-state regime-switching model (bull and bear regimes), a widely accepted framework for stock returns (\cite{Timmermann, Fruhwirth-Schnatter, Giner}), using a semi-Markov model to detect negative autocorrelations over very long time scales. He argues that, within this framework, the MACD indicator with three different time scales could be adapted to construct the optimal portfolio, {but obtaining statistically significant results remains difficult. So mean reversion at larger scales remains {in our opinion} unconvincing despite the fact it was already documented by \cite{Moskowitz}.} However, the two-state regime does not seem natural for modeling trend-following mechanisms, which are better explained by herding behavior rather than by an external guiding force. Models in which agents partially imitate each other are more likely to produce continuous regime changes rather than abrupt shifts between two states. Consequently, a more realistic approach would require a more complex model than a simple two-state regime. A more credible representation would assume that trends follow an Ornstein-Uhlenbeck process or a similar mean-reverting process.

%
{\clf \cite{Grebenkov} considered the case where the position depends linearly on the signal}\footnote{ {\cite{Grebenkov} most likely introduced  into the literature the "linear trading rule" whose position depends linearly on the indicator, rather than a purely binary trading rule, which had been the standard approach most likely because the standard approach focuses on a one-dimensional problem and on a strategy that targets a constant volatility. %
{\clf To the best of our knowledge, there is no prior published article {\clr except \cite{Brandt}, which is discussed in the appendix}; it is possible that we have missed some references.}  \cite{Grebenkov} introduced that linear methods in the literature because, at John Locke Investments that method was used to ensure that its statistical arbitrage trading program managing more than 500 dimensions or more than 500 single stocks maintained the property of ``rotational invariance.'' A linear trading rule was indeed necessary to have that rotational invariance property and design consistent strategies capable of performing trend-following either on an individual stock, or a basket of stocks, or on the eigenvectors of the correlation matrix of individual stocks, rather than independently applying trend-following to each stock in isolation.    }}%
{\clf  and showed that, for the process which follows a Brownian motion whose drift is an Ornstein--Uhlenbeck process, the $\EMA$ that maximizes the Sharpe ratio has a decay parameter $\eta_{\text{opt}}$ larger than $\lambda$, the decay parameter of the drift. At first sight this seems to disagree with \cite{Zakamuli}, for whom the optimal indicator has the same shape as the autoregressive coefficients of the return process, and we initially attributed the difference to the trading rule (linear rather than binary), since the two Sharpe formulas look different. This is not the case. As shown in Appendix \ref{comparaisonwithortherformulas}, in the one-dimensional Gaussian case without trading costs the two approaches are equivalent: both Sharpe ratios are increasing functions of the same correlation between the signal and the next return, so that  the binary and the linear rules share exactly the same optimal $\EMA$, their Sharpe ratios differing only by a factor close to $\sqrt{2/\pi}$. The apparent difference comes from the autoregressive coefficients of the process of \cite{Grebenkov}, which are less simple than they look: the weights $(1-\lambda)^{k}$ in Eq. \eqref{eq:autoco} apply to the unobserved shocks of the trend and not to past returns. Once the returns are written as a function of their own past, their autoregressive coefficients decay at the rate $\eta_{\text{opt}}$ and not $\lambda$ (Appendix \ref{kalmanlink}): $\eta_{\text{opt}}$ is the rate of the stationary Kalman filter of the model. The result of \cite{Grebenkov} is therefore consistent with \cite{Zakamuli}. The linear rule remains preferable for other reasons: under the Gaussian assumption its Sharpe ratio is higher, and it is in line with the ``rotational invariance'' property, the Markowitz optimization and the Agnostic Risk Parity of \cite{Benichou16}, where positions depend linearly on the signals.} {The linear trading rule can be extended from the one-dimensional case to the multi-dimensional setting, allowing the incorporation of the correlation matrix to better adjust portfolio positions in order to target the maximum Sharpe ratio at the portfolio level}.

%
{\clf The contribution of \cite{Grebenkov} lies elsewhere. To our knowledge, they were the first to give a fully explicit expression of the Sharpe ratio of a linear trend-following rule, when the asset price follows a Brownian motion whose drift is an Ornstein--Uhlenbeck process, as a function of the parameters of the return process ($\lambda$, $\beta_0$) and of the $\EMA$ ($\eta$), including trading costs through the turnover, together with the explicit optimal parameter and the distribution of the P\&L\footnote{The formula of \cite{Acar}, on which \cite{Zakamuli22} and \cite{Zakamuli} rely, gives the Sharpe ratio of a binary rule as a function of the correlation between the indicator and the returns, but not this correlation as a function of the parameters of a return process. When the correlation is computed for the process of \cite{Grebenkov}, both approaches lead to the same expression, of which the formula of \cite{Grebenkov} is the approximation for small $\lambda$, $\eta$ and $\beta_0^2$ (Appendix \ref{comparaisonwithortherformulas}).}.} That formula is key in our  main objective as we want to validate it using empirical measurements.

%
{\clf In the following years, the literature on optimal trend-following rules mostly kept to the one-dimensional case with binary or threshold positions}\footnote{{The formula in \cite{Acar, Zakamuli22,Zakamuli} {\clf is written for binary trading rules} (proof in Appendix \ref{proofhold}, {\clh originally} derived by \cite{Firoozye}). Following \cite{Grebenkov}, \cite{Ferreira18} published a similar formula for a linear trading rule, but in the context of a more general process with a simple moving average and not an $\EMA$.  \cite{Ferreira18} did not mention \cite{Acar} and derived his results through his own way. Later, after \cite{Grebenkov}, \cite{Firoozye} extended the formula of \cite{Acar} to linear trading rules%
{\clf . These extensions are not discussed in \cite{Zakamuli22} and \cite{Zakamuli}, which focus on binary rules}}}. {\clf \cite{Dai16} allowed only two possible positions (long or zero).} Similarly, \cite{Nguyen14b} determined the optimal threshold for triggering either a long or short position.

Our approach is statistical and complements a large literature on the economic
sources of return predictability in futures markets: carry \citep{Koijen2018},
momentum \citep{MiffreRallis2007,Moskowitz}, basis-momentum
\citep{BoonsPrado2019}, hedging pressure and the respective roles of hedgers
and speculators \citep{DeRoon2000,Kang2020}, and, at short horizons, the
reversal of the spread between adjacent maturities, which \cite{Rossi2026}
attribute to the different speeds at which information is incorporated along
the futures curve. These studies identify why returns are predictable. We take
the existence of trends as given, and ask how many time scales are needed to
describe them and which filter is optimal. The process of
Eq. \eqref{eq:autoco} is agnostic about the economic origin of trends, but two
main mechanisms can generate it. The first is the price impact of fundamental
traders, who incorporate information into prices only gradually as they build
their positions \citep{HongStein1999}. The second is herding: investors tend to
buy what has already risen, so that past returns feed future demand
\citep{Majewski,Valeyre}. Both produce a persistent drift that decays over
time, which is what Eq. \eqref{eq:autoco} describes.

Our first and {main} objective was to validate empirically the model of \cite{Grebenkov} through testing %
{\clf their formula} describing the sensitivity of the empirical Sharpe ratio of the ``Agnostic Risk Portfolio'' ($\ARP$)—one of the components of the optimal trend-following portfolio derived by \cite{Valeyre}—to the parameter of the $\EMA$. {This methodology ultimately extends the idea of \cite{Ferreira18} by focusing on an AR(1) model to capture trends and by aggregating all markets through projection onto the $\ARP$, rather than analyzing each market individually — which carries the risk of excessively noisy measurements and requires fitting complex and distinct ARIMA models}.  

{As a secondary result, we explored } more complex signals than $\EMA$ to confirm that one time scale is enough to describe trends. {\clb These empirical results challenge the common belief in our industry that complex systems are necessary. For example, \cite{Tzotechev} derived an optimization scheme for assigning weights across different time scales using a HRP-Markowitz framework, based on theoretical correlations results derived from an AR(1) process for returns. It also challenges %
{\clf the debated} \footnote{\cite{Buncic,Elder}} empirical and theoretical findings of \cite{Kelly}, who theoretically argue that simple models severely understate return predictability.} We decided to extend the model of \cite{Grebenkov} to Gaussian models with multiple time scales in Section \ref{secondtimescale}. For the other signals, however, we retained a combination of three $\EMA$s, in order to remain consistent with the $\MACD$. This choice may be justified by the possibility that the autoregressive structure of returns is more complex than that captured by a single $\EMA$. {However, there are two main differences. First, we aim to set the slope of the sensitivity at lag zero to zero, as our focus is on systematic returns through cross-asset futures, excluding specific risk. A positive slope would be undesirable, since it would imply that recent returns have less influence on the trend than older ones, whereas short-term mean reversion behavior is known to be weaker for indices than for single stocks. Second, we seek to obtain a larger fat tail at longer scales rather than a contrarian contribution, which may be more appropriate for capturing residual risk in single stocks but not for the systematic component of returns.}

{As another secondary result, we explored how} a simple Exponential Moving Average ($\EMA$), which is supposed to be both optimal and simple, could also be replicated by a highly complex but more usual combination of indicators, such as Bollinger Bands, whose elementary indicator is even nonlinear and path-dependent. This second {point} raises the question: why use a complex combination of complex indicators that are sensitive to cherry-picking when %
{\clf a simple solution} already exists?


\begin{figure}[H]
	\centering
	\begin{minipage}{0.45\textwidth}
		\centering
		\includegraphics[width=\textwidth]{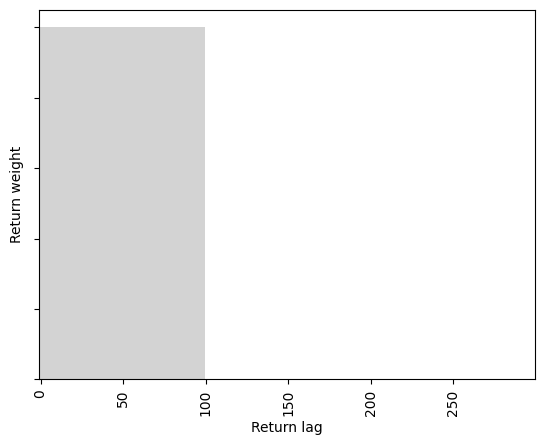}  
		\subcaption{MOM 100 days}
	\end{minipage} \hfill
	\begin{minipage}{0.45\textwidth}
		\centering
		\includegraphics[width=\textwidth]{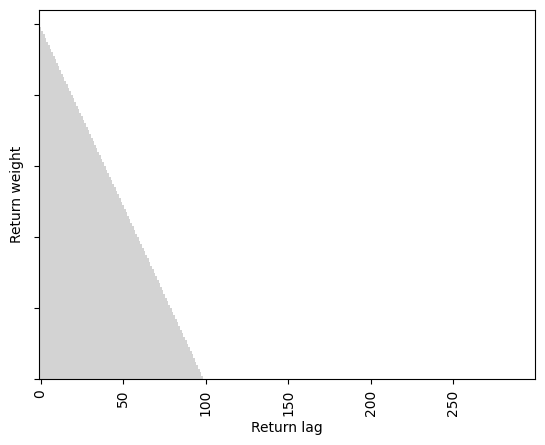}  
		\subcaption{SMA 100 days}
	\end{minipage}
	
	\vskip\baselineskip
	
	\begin{minipage}{0.45\textwidth}
		\centering
		\includegraphics[width=\textwidth]{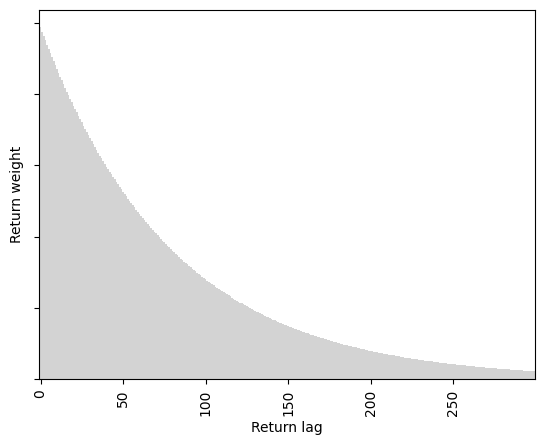}  
		\subcaption{EMA 100 days}
	\end{minipage} \hfill
	\begin{minipage}{0.45\textwidth}
		\centering
		\includegraphics[width=\textwidth]{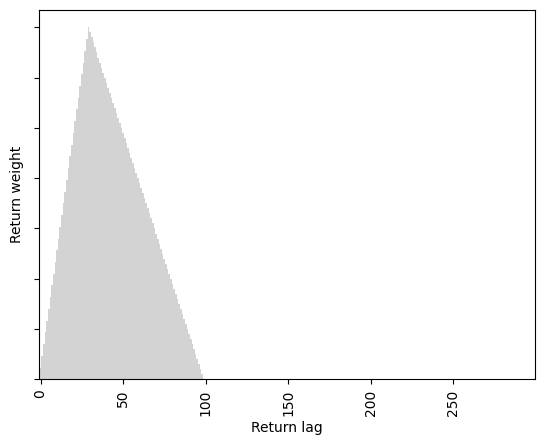}  
		\subcaption{SMA Crossover 30-100 days}
	\end{minipage}
	
	\vskip\baselineskip
	
	\begin{minipage}{0.45\textwidth}
		\centering
		\includegraphics[width=\textwidth]{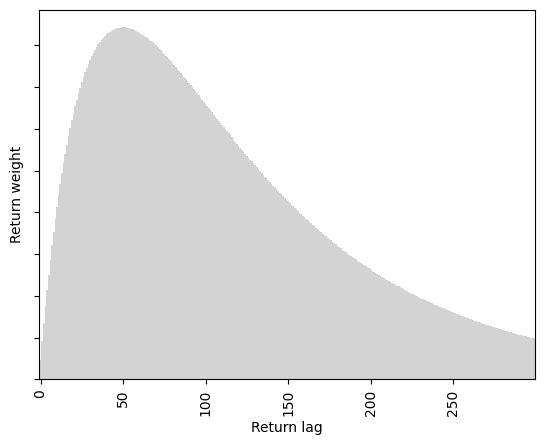}  
		\subcaption{EMA Crossover 30-100 days}
	\end{minipage} \hfill
	\begin{minipage}{0.45\textwidth}
		\centering
		\includegraphics[width=\textwidth]{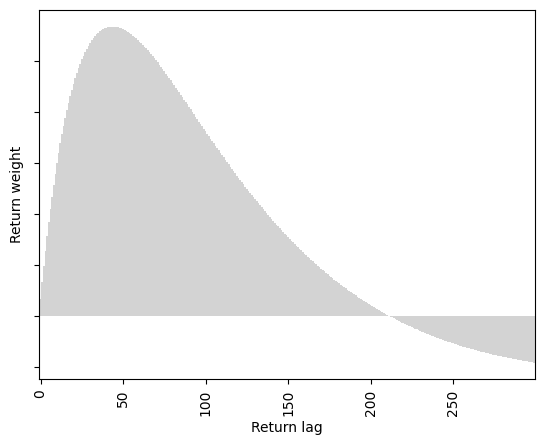}  
		\subcaption{MACD 30-100-400 days}
	\end{minipage}
	
	\caption{``Spectrum'' of different indicators inspired by \cite{Zakamuli} i.e. sensitivity to past daily returns. $\MOM$ {\clh is the difference between} two prices, $\SMA$ is the usual moving average on prices, $\EMA$ is the exponential {\clh moving average} of returns, $\SMA$ crossover is the difference between 2 $\SMA$, $\EMA$ crossover is the difference between 2 $\EMA$. $\MACD$ has 3 time scales and could be adapted according \cite{Zakamuli} which is contrarian on the large scales }
	\label{fig:sensitivity}
\end{figure}

\section{Theory from \cite{Grebenkov} }
\label{theory}

{To derive the theoretical Sharpe ratio for a trend-following strategy, we need a model that describes returns and trends. A classical one-dimensional model used to capture trends assumes that returns consist of two components: a random, independent and identically distributed (i.i.d.) component, and a small bias term. This bias can be modeled using an auto{-}regressive process of order 1, AR(1), which represents mean reversion process around zero with only one time scale. The AR(1) process is characterized by one time scale parameter—indicating the average time needed to revert to zero—and a standard deviation that determines the level of noise. AR(1) is very similar if time lag is infinitesimally small to the {\clh Ornstein-Uhlenbeck} process which is a  mean reversion usual and natural process in complex systems with only one time scale.}

{This classical one-dimensional model can be easily extended to multiple dimensions. In the multivariate case, two covariance matrices are required to account for correlations. The first matrix captures the correlations between the random components of the different underlyings, while the second matrix reflects the correlations among the AR(1) processes associated with each underlying. The model is formally described in Eq. \eqref{eq:autoco}, using the same notation as in \cite{Grebenkov}. In that equation the returns $r_{i,t}$ of the $i$th instrument among the N is described in two components: The first $\epsilon_{i,t}$ is a noise which is not autocorrelated while the second $\beta \sum_{k=1}^{t-1}{ \left(1-\lambda\right)^{t-1-k} \xi_{i,k}}$ is an AR(1) process which is autocorrelated where  $\xi_{i,k}$ is a noise which is not autocorrelated. Ultimately Eq. \eqref{eq:autoco} describes the returns of a diffusive process with a positive autocorrelation model for returns $r_{i,t}$: The weights $\left(1-\lambda\right)^{t-1-k}$ are declining exponentially. The parameter $\lambda$ describes the inverse of the relaxation time of the mean reversion AR(1) process describing the short bias, and $\beta$ represents the strength of that short stochastic bias in the returns. The noises $\xi_{i,k}$ are independent in the $k$  axis but may be correlated with $i$ axis (\cite{Valeyre} introduces specific forms of that matrix, leading to particular optimal portfolios). The noises $\epsilon_{i,t}$ are independent in the $t$ axis but may be correlated with $i$ axis with a very close correlation matrix between returns. Through that model, the trend is $\beta \sum_{k=1}^{t-1}{ \left(1-\lambda\right)^{t-1-k} \xi_{i,k}}$ and it follows an auto{-}regressive model of order 1, AR(1).}

\begin{equation} 
	r_{i,t}=\epsilon_{i,t}+ \beta \sum_{k=1}^{t-1}{ \left(1-\lambda\right)^{t-1-k} \xi_{i,k}}
	\label{eq:autoco}
\end{equation}

The original formulation of Grebenkov in Eq. \eqref{eq:autoco} can be rewritten in a more standard recursive form, commonly used in finance, by introducing a trend component $\mathrm{\mu}_{i,t}$

\begin{equation} r_{i,t} = \epsilon_{i,t} + \mathrm{\mu}_{i,t}, \end{equation} where the trend evolves according to \begin{equation} \mathrm{\mu}_{i,t} = \left(1-\lambda\right)\mathrm{\mu}_{i,t-1} + \beta \xi_{i,t-1}. \end{equation}


 {\clb Note that an alternative model is also used in the literature, for example in \cite{Tzotechev, Sepp,Zakamuli22}, where the returns, rather than the trends, are modeled directly using an AR(1) process or even an autoregressive fractal process.  {\clrr However, this approach has one fewer degree of freedom. The \cite{Grebenkov} model generates very noisy and not significant return Partial Autocorrelation Function, consistent with empirical observations. In the \cite{Grebenkov} model, return autocorrelations can remain very weak or even statistically insignificant despite the presence of trends, whereas in the models of \cite{Tzotechev, Sepp}, return autocorrelations are directly determined by the characteristic time scale. Another weakness of the alternative model is that the conditional expected return  is assumed to be measured with certainty from the $\EMA$ indicator, whereas with \cite{Grebenkov} it can only be measured with uncertainty and we will show later,  this distinction has important implications for portfolio construction. Finally, the model proposed by \cite{Grebenkov} belongs to the classical state-space framework in which the asset price follows a Brownian motion whose drift is an Ornstein--Uhlenbeck process. Such models are widely used in finance, applied mathematics, and Kalman filtering, for example in \cite{Ayed, Harvey, Kim, Rieder, Lakner}. Moreover, the model can be shown to be equivalent to %
{\clf an ARMA(1,1) process for returns}, although the interpretation of the resulting noise term is less straightforward.} }

Through variogram measurements on the Dow Jones Index, {\clrr as autocorrelation analysis was found to be excessively noisy}, \cite{Grebenkov} estimated that $\lambda=0.01$ and $\beta_0=0.1$ (Fig.\ref{fig:variogram}). {\clrr This approach also avoids relying on likelihood-based estimation, whose performance may deteriorate under model misspecification}. %
{\clf The quality of the fit is surprising}, considering the expectation of multiple time scales with long memory among investors {as well as the potential presence of mean-reversion patterns at longer horizons that are not captured}. Additionally, the result is not entirely convincing, as the 100-year period may be heterogeneous, and the fit is less accurate at shorter time scales, where measurements should, in principle, be less noisy. Another drawback is that the fit may be specific to the Dow Jones Index, the only index with such a long history. To say it in another way, we can suspect the presence of two distinct time scales, as the fit is not perfect for time scales shorter than 100 days. Another possible explanation is that the autocorrelation may not have been consistent over such a long period.

\begin{figure}[H]
	\centering
	\includegraphics[width=0.8\textwidth]{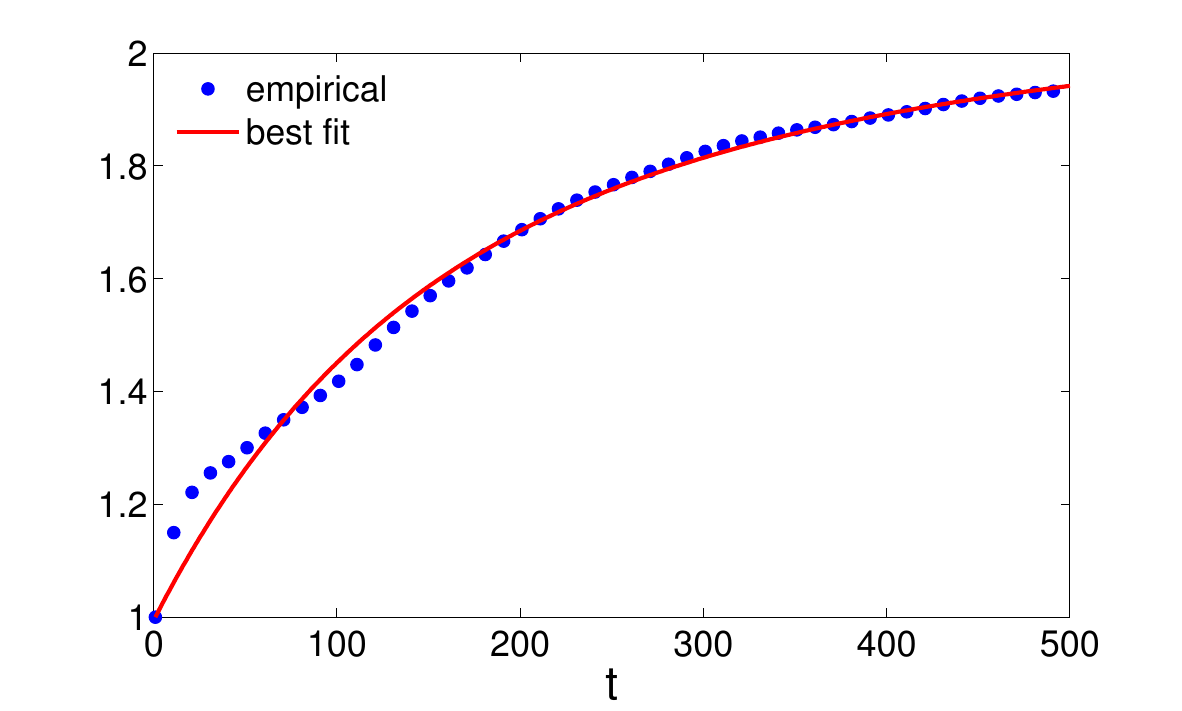}
	\caption{{\clb Variogram  of standardized logarithmic daily returns (normalized by realized volatility) to Dow Jones (1900-2012)} fit with  $\beta_0=0.08$ and $\lambda=0.011$ parameters, replicated from \cite{Grebenkov}  }
	\label{fig:variogram}
\end{figure}

%
{\clf \cite{Grebenkov} maximized the Sharpe ratio of a strategy whose positions depend linearly on the signal, assuming that returns are well described by Eq. \eqref{eq:autoco}.} {They considered only the one-dimensional case, but later, \cite{Grebenkov15} and \cite{Valeyre} extended the analysis to the multi-dimensional setting.}

%
{\clf \cite{Grebenkov} assume that positions depend linearly on the signals\footnote{They considered only the $\EMA$ as the signal. This is the right family: as \cite{Zakamuli} showed, the best predictor has the same functional form as the autoregressive coefficients of the returns, which decay exponentially for Eq. \eqref{eq:autoco} (Appendix \ref{kalmanlink}).}, in agreement with the Markowitz solution, whereas most of the literature considers positions that are long, short or flat with the same risk. They then determine the $\EMA$ parameter that maximizes the Sharpe ratio. In the one-dimensional Gaussian case this optimum coincides with that of the binary rule and with the best predictor; the benefit of the linear rule lies in the level of the Sharpe ratio and in the extension to several assets.}

{\clh The three relations of \cite{Grebenkov} used in this paper are the normalization of $\beta$ (Eq. \eqref{eq:optimalbeta0}), the theoretical Sharpe ratio (Eq. \eqref{eq:optimal}) and the optimal parameter of the $\EMA$ (Eq. \eqref{eq:optimal1}); they are discussed below.}

\begin{eqnarray} 
	\beta=\beta_0 \sqrt{\lambda \left(2-\lambda\right)} 
	\label{eq:optimalbeta0}\\
	\S\left(\eta\right)=\frac{\beta^2_0\sqrt{2\eta}-{\clh \frac{2}{\sqrt{\pi}}}\theta\sqrt{\eta}\left(\lambda+\eta\right)}{\sqrt{\left(\lambda+\eta\right)^2+2\beta^2_0\left(\lambda+\eta\right)}} 
	\label{eq:optimal}	\\
	\eta_{\text{opt}}=\lambda\sqrt{1+2\frac{\beta^2_0}{\lambda}}
	\label{eq:optimal1} 
\end{eqnarray}

In detail, they derived an explicit formula, in Eq. \eqref{eq:optimal}, the theoretical Sharpe ratio, $\S\left(\eta\right)$, when assuming returns are generated through their autocorrelation model in Eq. \eqref{eq:autoco}, as a function of $\eta$ (the smoothing parameter of the trend indicator $\EMA$ which is defined mathematically later in Eq. \eqref{eq:ema}) , $\lambda$ (the inverse of the time scale of the mean reversion process that describes the trends in Eq. \eqref{eq:autoco}) and $\beta$ (the weight of the trend component in the returns Eq. \eqref{eq:autoco}) or its normalized version $\beta_0$ while also considering the inclusion of trading costs, denoted as $\theta$. %
{\clf Eq. \eqref{eq:optimal} is the approximation, for $\lambda \ll 1$, $\eta \ll 1$ and $\beta_0^2 \ll 1$, of an exact expression given in Appendix \ref{linkwithgrebenkov}, where the approximations are listed; for our parameters the difference is between 1\% and 2\%.} As expected, the theoretical Sharpe ratio depends on $\beta$ (or $\beta_0$ its normalized version), the weight of the trend in the returns, in an approximately proportional manner when $\beta_0^2$ %
{\clf is significantly larger than $\lambda$} and when $\eta$ remains close to $\lambda$. In that case, the Sharpe ratio from Eq. \eqref{eq:optimal} can be approximated by $\beta_0 \sqrt{\frac{1}{2}}$ or $\beta_0 \sqrt{\frac{255}{2}}$ in the annualized version of the Sharpe. Fig.\ref{fig:optimal} %
{\clf illustrates Eq. \eqref{eq:optimal}} through a simple graph that helps to understand its sensitivity to the $\EMA$ parameter.

 In reality, Eq. \eqref{eq:optimal} was derived for the one-dimensional case only but it could be easily extended to the multi-dimensional case: {%
\cite{Grebenkov15} proved that Eq. \eqref{eq:optimal} holds when the correlations between $\xi_{i,k}$ and those between $\epsilon_{i,t}$ are set to zero, as shown in their Eq. (24)\footnote{Their Eq. (24) is equivalent{\clh , when $\nu=\kappa=1$,} to $\S=\sqrt{\frac{N q^2 (1-p^2)}{Q^2 + 2Q + R}}$ where $p = 1 - \eta$, $q = 1 - \lambda$, $Q = \frac{(1 - pq)\sigma^4}{\beta_0^4}$, and $R = 1 - q^2 - 2p^2q^2$}. Eq. \eqref{eq:optimal} also holds when the correlations between $\xi_{i,k}$ are the same as those between $\epsilon_{i,t}$, as shown in their Eq. (39)\footnote{Their Eq. (39) is equivalent to  $\S=\sqrt{\frac{N q^2 (1-p^2)}{Q^2 + 2Q + R}}$ where $p = 1 - \eta$, $q = 1 - \lambda$, $Q = \frac{(1 - pq)\sigma^4}{\beta_0^4}$, and $R = 1 - q^2 - 2p^2q^2$}, with both cases involving a scaling factor that depends only on the number of underlying assets. These latter conditions also make the Markowitz portfolio optimal, as proved by \cite{Valeyre}, but they are not realistic: the resulting “optimal Sharpe” would increase linearly with the square root of the number of underlyings, with the same slope regardless of correlation—an outcome that would be too ideal and profitable to be credible.  This in practice, with large $N$ yields an expected Sharpe ratio based on in-sample data that is much higher—by several orders of magnitude—than the out-of-sample Sharpe ratio. This is the signature of overfitting, which occurs simply because the model assumes that the estimated means are measured without error especially for small eigenvectors, thereby allowing unrealistically profitable hedges to reduce risk under the illusion of certainty about expected mean of returns.
 	
 	Eq. \eqref{eq:optimal} differs slightly when assuming that the correlations between $\xi_{i,k}$ are zero, as shown in their Eq. (43)\footnote{Their Eq. (43) is equivalent to $\S=\sqrt{\frac{N q^2 (1-p^2)}{Q^2 (1 - \bm{\rho_\epsilon}^2) + 2Q(1 - \bm{\rho_\epsilon}) + R}}$ where $p = 1 - \eta$, $q = 1 - \lambda$, and $Q = \frac{(1 - pq)\sigma^4}{\beta_0^4}$}, but this correlation-dependent formulation remains problematic: it fails to converge as $N \to \infty$, which contradicts empirical observations, suggesting that the underlying assumption (that correlations between $\xi_{i,k}$ are zero) is unrealistic.
 
 Nevertheless, we can reasonably expect Eq. \eqref{eq:optimal} to hold more precisely and perhaps even exactly—thanks to rotational invariance \footnote{We know from \cite{Grebenkov15} and the discussion above that Eq. \eqref{eq:optimal} holds when considering coordinate transformation to the eigenvectors of the correlation matrix of $\epsilon_{i,t}$, in which the "random" correlations between $\xi_{i,k}$ and those between $\epsilon_{i,t}$ would be set to zero under the $\ARP$ conditions.}—under the more realistic conditions that make the $\ARP$, and not the Markowitz portfolio, optimal.}

In that case where Eq. \eqref{eq:optimal} holds, Eq. \eqref{eq:optimalbeta0}
must be adjusted to redefine $\beta_0$. Eq. \eqref{eq:optimal} is derived for a
single asset. To apply it to a portfolio of $N$ assets, we assume that the
portfolio behaves as a single synthetic asset following Eq. \eqref{eq:autoco},
with the same $\lambda$ and an effective trend strength
\begin{equation}
	\beta_{0}=\frac{\beta}{\sqrt{\lambda\left(2-\lambda\right)}}\sqrt{\frac{N}{1+\left(N-1\right)\bm{\rho}^{2}}},
	\label{eq:beta0N}
\end{equation}
where $\bm{\rho}^2$ is the average squared correlation between the assets. The
factor $N/(1+(N-1)\bm{\rho}^2)=(\mathrm{tr}\,C)^2/\mathrm{tr}\,C^2$ is the
effective number of independent assets of the universe. Rescaling $\beta_0$,
rather than the Sharpe ratio, keeps two parameters to fit in
Eq. \eqref{eq:optimal} instead of three. This assumption is not derived in
general. It is supported by the case where the trend is carried by a single
direction $u$ of the universe, i.e. where $u$ is the only significant
eigenvector of the correlation matrix of the $\xi_{i,k}$ of
Eq. \eqref{eq:autoco}, an assumption introduced in \cite{Valeyre} to justify the
optimality of the $\ARP$. Trend following on that direction is then exactly a
one-dimensional problem, in which the trends of the $N$ assets add up
coherently whereas their noises do not: $\beta_0^2$ is multiplied by
$N/(u^{\dagger}Cu)$, and the optimal $\EMA$ is faster, as expected for a Kalman
filter that receives more information. 

  {The rescaling factor formula, Eq. \eqref{eq:beta0N}, is empirically tested through the relationship between the Sharpe ratio and $N$ in Section \ref{scaling}. This rescaling factor formula also corresponds to the theoretical scaling of the Sharpe ratio for a portfolio of $N$ assets with identical expected returns and volatilities, and an average pairwise correlation of $\bm{\rho}^2$. This makes sense, as the correlation between two trend-following strategies built on underlyings correlated at $\bm{\rho}$ is expected to be $\bm{\rho}^2$. Thus, the scaling factor for $\ARP$ ultimately almost coincides with that obtained for the naïve '1/N' trend-following portfolio, which does not use the correlation matrix to determine positions. The subtle difference is that the scaling inside $\beta_0$ not only rescales the Sharpe ratio in the same way as for the naive '1/N' portfolio around $\eta=\lambda$, but also slightly shifts the optimal smoothing parameter $\eta_{\text{opt}}$. }

{Additionally \cite{Grebenkov} derived  in Eq. \eqref{eq:optimal1},  the optimal smoothing parameter ($\eta_{\text{opt}}$) of the Exponential Moving Average ($\EMA$) applied to returns as a trend indicator for a trend-following strategy which is supposed to yield to the optimal Sharpe ratio}. It is interesting to observe from Eq. \eqref{eq:optimal1} that $\frac{1}{\eta_{\text{opt}}}$, the time scale of the signal $\EMA$ generating the optimal strategy is always shorter than $\frac{1}{\lambda}$, %
{\clf the time scale of the trend and of the autocorrelation of returns. This is not specific to the linear rule: $\eta_{\text{opt}}$ is the rate of the stationary Kalman filter of the model of Eq. \eqref{eq:autoco}, i.e. of the best predictor of the next return given past returns, whose weights are the autoregressive coefficients of the returns (Appendix \ref{kalmanlink}). Since the trend is observed only through noisy returns, the optimal filter balances the averaging of the noise against the tracking of a trend that keeps moving: $\eta_{\text{opt}}$ tends to $\lambda$ when $\beta_0 \to 0$, and when $\beta_0$ increases (either due to a stronger trend or a more diversified universe), recent returns become more informative and the time scale of the optimal $\EMA$ is reduced.} As a conclusion, theoretically the parameters of the signals should be adjusted when the universe is increased.




\begin{figure} [H]
	\centering
	\includegraphics[width=0.8\textwidth]{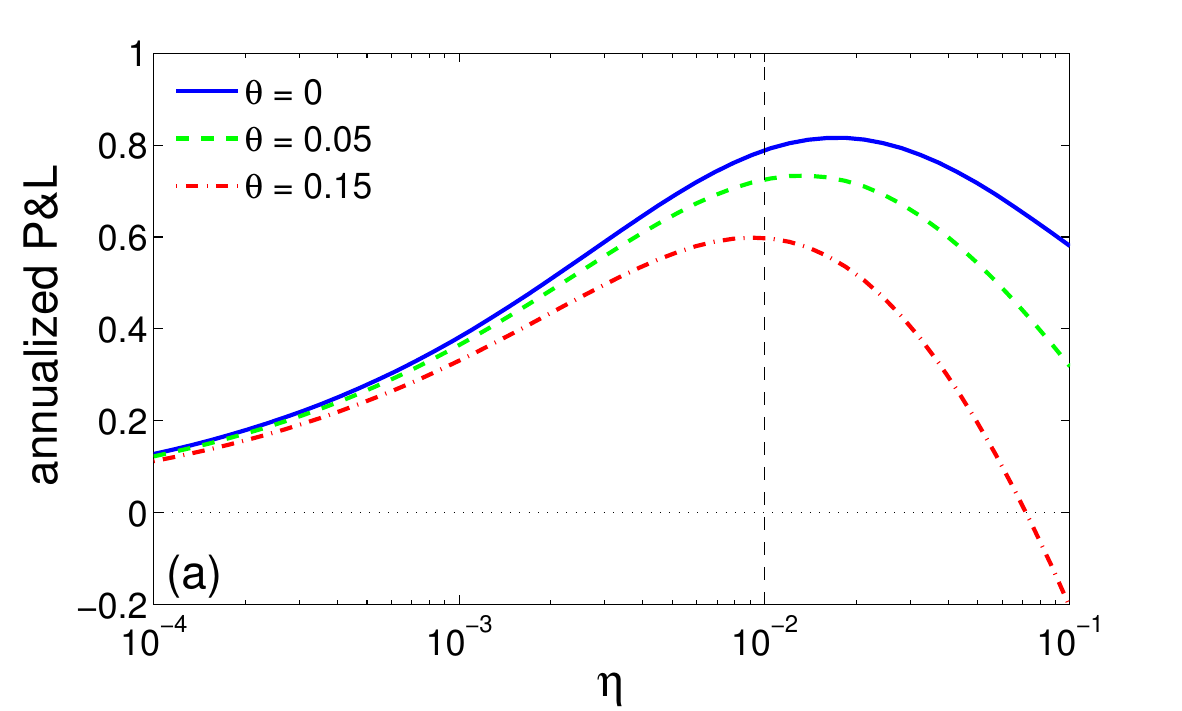}
	\caption{Theoretical Sharpe ($\S\left(\eta\right)$ from Eq. \eqref{eq:optimal}) of the trend-following strategy depending on $\eta$ the parameter of the $\EMA$. $\beta_0=0.1$ and $\lambda=0.01$ for different cost of trading $\theta$, replicated from \cite{Grebenkov} {\clh (panel (a) of their Fig. 4; the vertical axis, labelled ``annualized P\&L'' in the original figure, is the annualized risk adjusted P\&L, i.e. the Sharpe ratio)}} 
	\label{fig:optimal}
\end{figure}


  \section{Description of the empirical analysis}
  
  
{To validate Grebenkov’s theory in Eq. \eqref{eq:optimal}, introduced in Sec.\ref{theory} and, therefore, the diffusive process in Eq. \eqref{eq:autoco} which is assumed for the underlying assets, we conducted empirical backtests of the Agnostic Risk Parity ($\ARP$) strategy introduced by \cite{Benichou16}. {We decided to use the $\ARP$ as the projection portfolio, as it is theoretically and empirically considered to be the optimal trend-following portfolio, capable of minimizing measurement noise in the Sharpe ratio. However, we expect that the results should also hold even with lower Sharpe ratios and more noises for more conventional and naive '1/N' trend-following portfolios, such as those that do not manage correlations or assume an identity correlation matrix except when computing volatility used to target the size of the portfolio. }

	{\clr The other objective advantage of $\ARP$ is that it} constructs a rotationally invariant portfolio {which makes the Sharpe ratio less dependent on the subjective selection of futures used to define the investment universe, as well as on discretionary choices regarding the relative number of FX, stock index, bond, and commodity futures. $\ARP$ also by consequence }applies trend-following to {\clh eigenvectors} while allocating equal unconditional risk to each eigenvector of the covariance matrix. While rotational invariance is a desirable theoretical property, the equal unconditional risk allocation among eigenvectors is a more desired property for asset manager{s}. Moreover, \cite{Valeyre} demonstrated that the $\ARP$ portfolio is optimal—achieving the highest Sharpe ratio—under the assumption that the covariance of the $\xi_{i,k}$ terms in Eq. \eqref{eq:autoco} follows a particular random matrix structure, which is considered realistic. Implementing the $\ARP$ portfolio is relatively straightforward and consists of normalizing the trend indicators by the inverse square root of the correlation matrix, rather than by the full inverse as in the classical Markowitz framework.}
  
 {We decided to focus primarily on exponential moving averages ($\EMA$) as trend indicators to feed the Agnostic Risk Parity ($\ARP$) strategy, as they are theoretically optimal according to \cite{Grebenkov, Zakamuli}, assuming the validity of Eq. \eqref{eq:autoco}. {\clrr The formula in Eq. \eqref{eq:optimal} {\clh is moreover valid} only under the assumption that the $\EMA$ is used as the indicator. There is no tautology. If we replicate Grebenkov’s formula of Eq. \eqref{eq:optimal} using an $\EMA$, this indicates that the Grebenkov model is appropriate. Indeed, we would not replicate Grebenkov’s formula for another process, even when using an $\EMA$ as the indicator}.  Accordingly, we tested various $\EMA$ configurations and measured how the empirical Sharpe ratio of the resulting portfolio depends on the $\EMA$ parameter, in line with Grebenkov’s theoretical predictions. Additionally, we explored more complex trend indicators based on multiple time scales—such as the Moving Average Convergence Divergence ($\MACD$)—to empirically validate whether the $\EMA$ remains the optimal choice.}
  
  {We backtested the Agnostic Risk Parity ($\ARP$) strategy on a global universe including Futures on commodities, FX, stock indices and bonds on a period 25th of {M}ay 1990 and stops on the 7th of {D}ecember 2023. } 
   
   \subsection{ {Renormalized exponential moving average indicator } ($\EMA$)}
   \label{rernormalizedema}
   
  { We used only a normalized version of the $\EMA$ to account for the dependence of volatility on both time and instrument. This normalization allows for meaningful comparisons of trend indicators across instruments, time periods, and even different smoothing parameters. At time $t$, the $\EMA$ is a vector of $N$ exponential moving averages of normalized returns, one for each of the $N$ instruments. The $\EMA$ of instrument $i$ at time $t+1$ is defined by $\phi_{i,t}$ in Eq. \eqref{eq:ema}, where $r_{i,t}$ denotes the logarithm return of instrument $i$ at time $t$. The parameter $\sigma_{i,t}$ represents the standard deviation of returns used to normalize them, and $\eta$ is the smoothing parameter of the $\EMA$ as well as of the return normalization. $\sigma^2$ can be interpreted as an exponential moving average of squared returns, and the indicator $\phi$ corresponds to an exponential moving average of returns scaled by their volatility.} The incrementation  is on {a }daily basis but it could be generalized to minutes returns. 
  
 {\clr	It is interesting to note that the half-life (originally introduced in nuclear physics), defined as the lookback period over which the cumulative weight of past observations reaches $50\%$, corresponds approximately to $\frac{\ln(2)}{\eta}$. That value could be compared to the half of the look-back period of the momentum indicator, which is a more standard and usual indicator in the financial industry. 
 By contrast, the standard number of periods used in technical analysis and in financial markets for the exponential moving average, as defined in \cite{Wilder, Murphy}, is $\frac{2}{\eta}$. 
 This latter quantity corresponds to an average lookback period, but it is not well suited to use a weighted arithmetic average of periods, and may be a weighted harmonic average would be better here. As a result, it has weaker economic interpretation and cannot be meaningfully compared with the parameters of other indicators.
 }
   
   \begin{eqnarray}
   	\sigma^2_{i,t+1}\left(\eta\right) =  \left(  1-\eta \right) {\clh \sigma^2_{i,t}}\left(\eta\right)+ \eta r^2_{i,t+1} 
 \\
   	\phi_{i,t+1}\left(\eta\right)=\left(\ 1-\eta \right) \phi_{i,t}\left(\eta\right)+ \sqrt{\eta} \frac{r_{i,t+1}}{\sigma_{i,t}\left(\eta\right)} 
   	\label{eq:ema}  	
   \end{eqnarray}
   
   The indicator $\EMA\left(150\right)$ is defined by the $\phi_{i,t+1}\left(\eta\right)$ when $\eta=\frac{1}{150}$ is applied to every underlying $i$ and time $t$. We use  $\sqrt{\eta}$ in Eq. \eqref{eq:ema} so that the std of $\phi_{i,t+1}\left(\eta\right)$ is 1 theoretically { if $r$ are well approximated by unautocorrelated returns so that $\phi$ is a normalized $\EMA$ indicator}.
   
   \subsection{{ Moving Average Convergence Divergence indicator ($\MACD$) } as a 3  time scales $\EMA$ indicator}
   
   Inspired by $\MACD$, we introduced a combination of $\EMA$ applying the Eq. \eqref{eq:threeema} while determining $\omega_1$ so that the derivative of the sensitivity to past daily returns  is at 0 at the lag 0 and replicate curves in Fig.\ref{fig:oursensitivity}. {The idea behind this indicator was to reduce the weight of returns at very short time scales—where mean reversion is possible—and to increase the weight of returns at longer time scales, in order to test a multi-timescale memory effect for trend-followers.} 
   \begin{equation}
    \MACD_{i,t}\left(\eta_1,\eta_2,\eta_3,\omega_1,\omega_2,\omega_3\right)=\omega_1 \phi_{i,t} \left(\eta_1\right)+\omega_2 \phi_{i,t} \left(\eta_2\right)+\omega_3 \phi_{i,t} \left(\eta_3\right)
    \label{eq:threeema}
    \end{equation}
    
    The derivative at zeros yields to Eq. \eqref{eq:threeema1}.
     \begin{equation}
    	0=\omega_1 \sqrt{ \eta_1}+\omega_2 \sqrt{ \eta_2}+\omega_3 \sqrt{\eta_3}
    	\label{eq:threeema1}
    \end{equation}

 \subsection{Agnostic Risk Parity ($\ARP$) and naive '1/N' portfolio}   
 
 

 We first estimated the correlation matrix $C$, of dimension $N \times N$, using a 750-day exponential moving average applied to weekly returns when implementing the RIE filter introduced by \cite{Bun16}. The vector $\Sigma$ consists of N values representing the standard deviations estimated using a 40-day exponential moving average applied to daily returns.
 
 Next we expressed the positions vector as a linear function of the signals ($\EMA$ through $\phi$ or $\MACD$) applying a normalisation that involves the inverse of the square root of the estimated correlation matrix and volatilities. This follows the formulation of the agnostic risk parity portfolio (${\ARP}$) introduced by \cite{Benichou16}.
 
 We then used the parameter $\rho = \frac{1}{20}$ for portfolio smoothing as specified in  Eq. \eqref{eq:arp2} which is {\clh a very easy} solution to reduce trading cost {(usually, trend indicators include rules to reduce trading costs. While the $\EMA$ does not explicitly aim to avoid excessive trading, smoothing the portfolio has proven to be both simple and effective. \cite{Benichou16} recommended optimization under constraints {to limit trading and illiquidity and} to better manage trading costs)}.
 
Finally, we applied a resizing process in Eq. \eqref{eq:arp} to target a constant volatility for the final positions ${\ARP}$ which is a vector of N weights at time $t+1$. {The target-volatility step is both usual and useful to get standardized and homoscedastic returns for the portfolio so that Sharpe ratio is adapted}.

 {We applied exactly the same computation for the naive '1/N' approach, except that we assumed $C^{-0.5}$ in Eq. \eqref{eq:arp1} to be the identity matrix. This is equivalent to managing trend-following strategies independently, underlying by underlying, while constraining the overall portfolio to maintain constant volatility. This naive approach is referred to as ‘1/N’ by \cite{Benichou16}. %
{\clh We did not test the true Markowitz portfolio, which would simply replace $C^{-0.5}$ by $C^{-1}$, as it is well known among practitioners that this solution performs significantly worse.} We also test the naive '1/N' approach derived from the binary trading rule where we replace $\phi_{t}$ by its sign. }

 \begin{eqnarray}
 	\left\{
 	\begin{array}{l}
 	\hat{\ARP}_{t+1} = \left(  1-\rho \right)   \hat{\ARP}_{t} +\rho \Sigma^{-1}_{t} C^{-0.5} \phi_{t}\left(\eta\right) \\[10pt] or \\[10pt]
 		\label{eq:arp1}
 	\hat{\ARP}_{t+1} = \left(  1-\rho \right)   \hat{\ARP}_{t} +\rho \Sigma^{-1}_{t} C^{-0.5} \MACD_{t}\left(\eta_1,\eta_2,\eta_3,\omega_1,\omega_2,\omega_3\right) \\
 		\label{eq:arp2}
 	\end{array}
 	\right. \\
 	{\ARP}_{t+1}=\frac{\hat{\ARP}_{t+1}}{\sqrt{\hat{\ARP}_{t+1}'{\Sigma_{t}}C{\Sigma_{t}} \hat{\ARP}_{t+1} }} 	
 		\label{eq:arp}
 \end{eqnarray}
 
  \subsection{Dataset and different simulated parameters}   

 The simulation starts on the 2{9}th of {\clb M}ay 1990 and stops on the 7th of {\clb D}ecember 2023. We used daily returns from 73 futures instruments in stock indices, bonds, FX and commodities futures. The description is in the appendix \ref{Data}.
 
 We tested the different indicators applying the $\ARP$ formula Eq. \eqref{eq:arp}. The different parameters are described in Tab.\ref{tab:tab0} in the appendix \ref{Parameters}.

Fig.\ref{fig:oursensitivity} displays  the sensitivities of these indicators to past daily returns. $\MACD$ enables {us} to put more weights on very older returns as we expected. 
 

 \begin{figure}[H]
 	\centering
 	\begin{minipage}{0.45\textwidth}
 		\centering
 		\includegraphics[width=\textwidth]{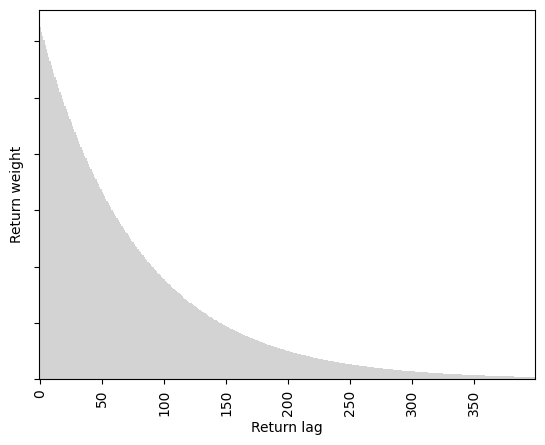}  
 		\subcaption{$\ARP\left(80\right)$}
 	\end{minipage} \hfill
 	\begin{minipage}{0.45\textwidth}
 		\centering
 		\includegraphics[width=\textwidth]{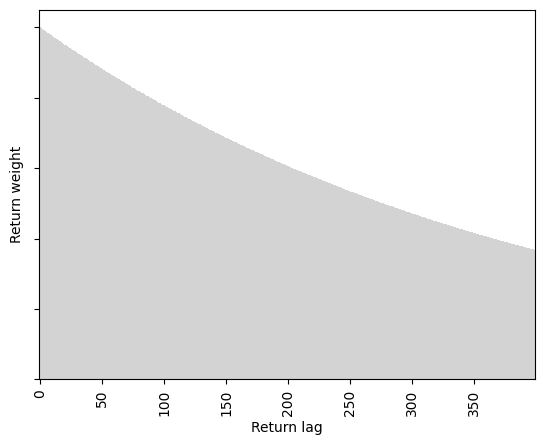}  
 		\subcaption{$\ARP\left(400\right)$}
 	\end{minipage}
 	
 	\vskip\baselineskip
 	
 	\begin{minipage}{0.45\textwidth}
 		\centering
 		\includegraphics[width=\textwidth]{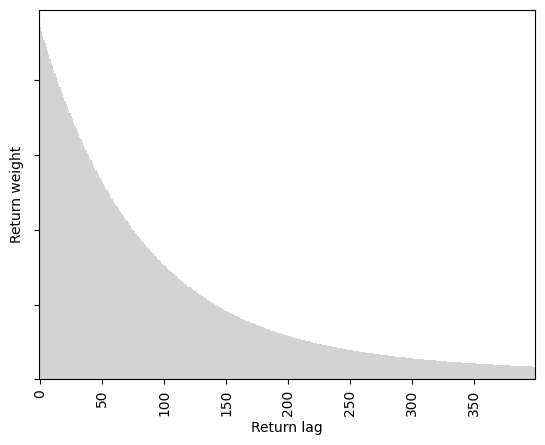}  
 		\subcaption{ $\ARP\left(0\times20, 80,0.2\times400\right)$}
 	\end{minipage} \hfill
 	\begin{minipage}{0.45\textwidth}
 		\centering
 		\includegraphics[width=\textwidth]{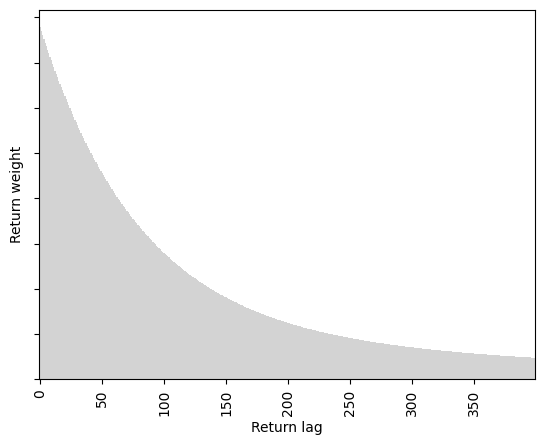}  
 		\subcaption{ $\ARP\left(0\times20, 80,0.4\times400\right)$}
 	\end{minipage}
 	
 	\vskip\baselineskip
 	
 	\begin{minipage}{0.45\textwidth}
 		\centering
 		\includegraphics[width=\textwidth]{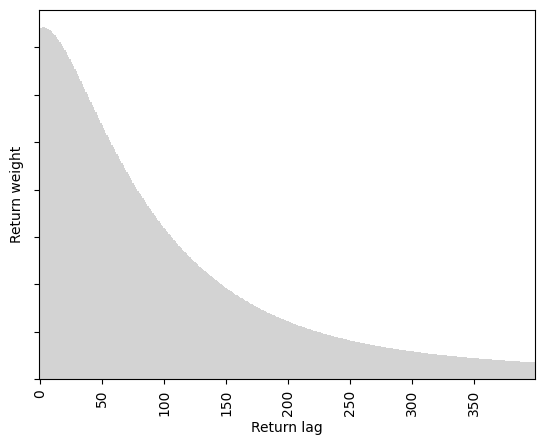}  
 		\subcaption{ $\ARP\left(20, 80,0.2\times400\right)$}
 	\end{minipage} \hfill
 	\begin{minipage}{0.45\textwidth}
 		\centering
 		\includegraphics[width=\textwidth]{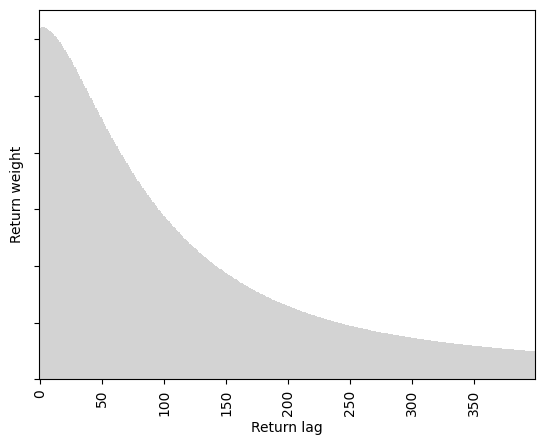}  
 		\subcaption{ $\ARP\left(20, 80,0.4\times400\right)$}
 	\end{minipage}
 	
 	\caption{Different sensitivities to past daily returns for indicators ($\MACD$, $\EMA$) all based on 20, 80 and 400 day time scales. we have $\ARP\left(0\times20, 80,0.2\times400\right)$, $\ARP\left(0\times20, 80,0.4\times400\right)$, $\ARP\left(20, 80,0.2\times400\right)$, $\ARP\left(20, 80,0.4\times400\right)$, $\ARP\left(80\right)$, $\ARP\left(400\right)$  }
 	\label{fig:oursensitivity}
 \end{figure}

 \section{Empirical results}
 
 Here we first interpret the simulated empirical Sharpe ratio when applying the $\ARP$ portfolio with the different indicators. %
{\clf The results challenge traditional recipes:} {the empirical Sharpe ratio closely replicates the theoretical Sharpe ratio predicted by \cite{Grebenkov}, whose main result shows that a single $\EMA$ indicator is already optimal. {\clr Secondly, we validate empirically the rescaling formula}. Thirdly, we attempt to replicate the performance of a simple $\EMA$ using a basket of more complex trend indicators. Finally, we incorporate the non-linear findings of \cite{Schmidhuber}, which could challenge our assumptions, and we analyze the empirical correlations between different strategies.}
 
 \subsection {Grebenkov's model empirical validation }
 
  {In Section \ref{intro} and \ref{theory}, we introduced both the empirical and theoretical Sharpe ratios. As a reminder, the empirical Sharpe ratio is derived from backtesting, whereas the theoretical Sharpe ratio is obtained by modeling the underlying asset using theoretical diffusive processes. We tested Eq. \eqref{eq:optimal}, which describes the theoretical Sharpe ratio as a function of the parameters of the trend indicator, as derived by \cite{Grebenkov}.}
 
 Interestingly, the empirical simulation for $\ARP$ %
{\clf fits well} the {\clh theoretical formula} of the Sharpe ratio Eq. \eqref{eq:optimal} derived in \cite{Grebenkov} with the following parameters {$\lambda=\frac{1}{180\pm 17}$} and $\beta_0=0.12$. {Using the \texttt{scipy.optimize} library, we obtained precisely $\lambda = \frac{1}{180.65}$, with an indicative 95\% confidence interval of approximately $\left[\frac{1}{209}, \frac{1}{158}\right]$. This interval is only indicative due to strong nonlinearities and significant autocorrelations in the data. Using bootstrap we obtain a similar interval $\left[\frac{1}{223}, \frac{1}{157}\right]$. We also obtained a $R^2=0.98$ which could be only an indication for formal goodness-of-fit test}.  %
{\clf Fig.\ref{fig:empirical_theory} shows that the empirical measurements are close to the theoretical curve, which supports Grebenkov's model to describe trends.}

 The nature of this validation should be stated precisely. The two parameters
 $\lambda$ and $\beta_0$ are calibrated on the full 1990--2023 sample: the fit of
 Fig. \ref{fig:empirical_theory} is therefore, strictly speaking, an in-sample calibration of
 two parameters on nine measurements, and the $R^2$ and the confidence intervals
 reported above are indicative rather than a formal goodness-of-fit test. What is
 \emph{not} fitted is the specification itself. The return process of
 %
{\clh Eq. \eqref{eq:autoco}} and the functional form of Eq. \eqref{eq:optimal}
 were published in 2014, before and independently of this study, and were
 calibrated by their authors on a single instrument (the Dow Jones index,
 1900--2012) with a different estimator (the variogram of returns). None of the
 data used here (73 futures on commodities, currencies, bonds and stock indices)
 played any role in the specification of the model, and we did not try and
 discard alternative functional forms. The risk of overfitting through model
 selection, which is the main source of overfitting in backtests, is therefore
 limited, and the test is close to an out-of-sample test of the
 specification\footnote{We deliberately did not split our sample into an
 	estimation period and a hold-out period, for two reasons. First, $\beta_0$
 	declines over time as market inefficiencies fade: such a split would mainly test
 	the stationarity of $\beta_0$, which the model does not claim, while making
 	these already difficult measurements considerably noisier ( illustration  in the Section \ref{decline} of the Appendix).
 	Second, and more fundamentally, a hold-out period drawn from historical data
 	that were already available to the researcher is not regarded as out-of-sample
 	by practitioners: the only test that is fully protected against selection is
 	the performance observed after the model has been fixed, i.e.\ in production.
 	The post-publication decay of many documented anomalies supports this caution
 	\citep{McLeanPontiff2016,HarveyLiuZhu2016}. Our setting approaches this standard
 	for the specification, which was fixed and published before our measurements;
 	it does not reach it for the parameters, which is why we describe the fit as
 	``almost'' out-of-sample. We are aware that this is not the standard academic
 	protocol; it is the one that carries weight in the industry, and one aim of
 	this paper is to bridge the two.}.
 	
 	 A good fit, however, does not by itself rule out a richer description of
 	the trends. We therefore test in Sect. \ref{secondtimescale} of the Appendix whether the
 	measurements call for a second time scale. The fitted model
 	is exactly the one-scale model. The data therefore do not call for a
 	second time scale.

 Our estimates are of the same order of magnitude as those of \cite{Grebenkov},
 but they are not identical: we obtain $\lambda\simeq 1/180$ against
 $\lambda=0.011\simeq 1/90$, which lies outside our confidence interval, and
 $\beta_0=0.12$ against $0.08$. Such differences are expected: the instruments,
 the period and the estimator all differ, $\beta_0$ depends on the size of the
 universe, and the variogram fit of \cite{Grebenkov}
 is imperfect at short lags.

 
 {We can note that the measurements using $\ARP$ fit the theoretical results better than those obtained with the naïve '1/N' portfolio (using the \texttt{scipy.optimize} library, we obtained for the naive '1/N' case $\lambda=1/110$ but $R^2=0.75$), where the correlation matrix is not considered to optimize diversification and reduce noise. The naïve '1/N' measurements are not only shifted downward  as we could have expected %
{\clf (the shift would have been larger with a less diversified universe)}, but also appear slightly rotated (with $\lambda$ slightly increased) — a difference that may simply result from noisier estimates, or from the assumption made in Section \ref{theory} regarding the extension of Grebenkov’s formula, using the same $\lambda$, from the one-dimensional to the multi-dimensional case. It may be that Grebenkov’s formula holds more rigorously, from a mathematical standpoint, for the optimal $\ARP$ portfolio than for the naïve '1/N' one. However, this would be quite surprising, since Grebenkov’s formula is derived precisely under the mathematical conditions that make the naïve '1/N' approach optimal (correlations between $\xi_{i,k}$ and those between $\epsilon_{i,t}$ are set to zero), as detailed in Section \ref{theory}. The most likely explanation, which is not very intuitive since the error in the graph is highly autocorrelated with an estimated std of error of 0.18 in line with the 30 years of history, is that the $\ARP$ measurements effectively apply a similar process to the  weighted average of $\lambda$ per market, assigning higher weights to markets which are different from other and which are not simply repetitions of others (as in the stock indices cases where all indices are very close). In contrast, the naïve '1/N' approach uses a similar process to an equal-weight scheme, which can overweight markets like stock indices that repeat themselves, thereby amplifying measurement errors. As a result, the naïve '1/N' portfolio has just overweighted markets which appear faster. \cite{Schmidhuber} similarly aggregates markets to obtain more robust results, observing that these coefficients are universal across asset classes and exhibit a universal scaling behavior, as the trend’s time horizon ranges from a few days to several years. We can also note that the $\ARP$ approach yields higher Sharpe ratios than the naïve '1/N' approach, even when binary trading rules are used, the overperformance would have been a lot larger if the universe was not as diversified.}%

The binary rule was applied to the naive '1/N' portfolio only, without the
$\ARP$. In the one-dimensional Gaussian case, the Sharpe ratio of the binary
rule should be close to $\sqrt{2/\pi}\approx 0.80$ times that of the linear rule
(Eq. \eqref{ratio_binary_linear} in Appendix \ref{linkwithgrebenkov}). This is
not what Fig.\ref{fig:empirical_theory} shows for the naive portfolio: for $\EMA$
parameters between 50 and 180 days, the binary rule reaches Sharpe ratios
comparable to those of the linear rule. This does not contradict the
one-dimensional result: we attribute it to a diversification effect which is
specific to the naive '1/N' portfolio. With the linear rule, the risk allocated
to each underlying is proportional to the magnitude of its signal, so that the
naive portfolio concentrates its risk on the few underlyings, often correlated
with each other, whose trends are the strongest at a given time. The binary rule
allocates the same risk to every underlying, whatever the magnitude of its
signal, and is therefore better diversified, which offsets the loss expected in
one dimension. The $\ARP$ manages this diversification directly through the
correlation matrix, which is why it yields higher Sharpe ratios than the naive
portfolio with either rule.

 As a first consequence the parameter of $72$ business days ($\eta_{\text{opt}}=\frac{1}{72}$ {\clr or a half life of 50 business days}) for simple $\EMA$ is the optimal parameter to get the optimal Sharpe ratio when not accounting cost of trading which are very small at that trading frequency ({\clb Tab.\ref{tab:tab1} displays holding periods slower than 80 days}). \cite{Lamperiere} found a different result, with a slower optimal parameter at 5 months {\clr as half life for the $\EMA$}, but very similar Sharpe ratios across $\EMA$s ranging from 1 to {\clr10} months (see their Tab.1), showing an almost flat curve with no clear optimum. The difference with our results may be explained by the choice made in \cite{Lamperiere} to determine positions based not on the linear magnitude of the $\EMA$ signals, but solely on their sign. Additionally, they {use the naive '1/N' approach and} did not use the $\ARP$ portfolio construction, which involves inverting the square root of the correlation matrix {which slightly improves the fit, as observed in Fig.\ref{fig:empirical_theory}.} {\clr They also use a monthly rebalancing and a "de-biased" trend. Their Tab.1 results are based on the period 1960-2012. They also use, as indicator, the exponential moving average of price differences (computed at a monthly frequency), divided by the exponential moving average of absolute monthly price changes, which is close but not strictly equivalent to the EMA of normalized daily returns as defined in Eq. \eqref{eq:ema}. } For these different reasons, their Sharpe ratio for the post-2000 period was measured at 0.85, lower than our 1.2, and their optimal $\EMA$ signal %
{\clh corresponded to a slower timescale than ours.} {\clrr Note that Table 2 in \cite{Moskowitz} exhibits the same curve shape as ours, with the maximum achieved for a momentum lookback period of approximately 9 to 12 months.}
 
 As a second consequence of the very good fit with the theoretical formula,  modeling trend through mean-reversion process using only one relaxation time and not an multi-time scales one appears to be a good solution %
{\clf as the fit is good}. That is {\clh particularly} unexpected as market is known to have a multi{-}time scales property: For example the relaxation of volatility is known to have mu{\clb l}ti-time scales,   investors are expected to have different horizons of time and different horizons of analysis. As a consequence we can wonder w{\clb h}ether the usual recipe to take into account a multitude of different indicators is justified to claim having the most robust and the most refined signal. {\clb When properly renormalizing the signal, the position, and the portfolio, we do not observe trends reversing over longer horizons (e.g., 12 months) that would indicate mean reversion. Instead, the Grebenkov formula still holds even at horizons of up to 1,000 days. This stands in strong contradiction to the findings of \cite{Moskowitz}, who focused on momentum ($\MOM$) indicators over the period 1965–2009.} {\clrr %
{\clf Their interpretation relies on autocorrelation measurements} \footnote{For example, in Fig. 1 of \cite{Moskowitz}, most of the *t*-statistics for lags beyond 10 months in the regression of monthly returns on past momentum are close to (-1) in average with values from 2 to -3. However, it is unclear whether the authors accounted for the correlation structure in the data properly (they said "Stacking all futures contracts and dates,we run a pooled panel regression and compute t-statistics that account for group-wise clustering by time"), which could further reduce the statistical significance of these estimates. Moreover, the results differ across asset classes. Finally, there is substantial autocorrelation across the reported estimates because momentum signals computed at different lags are themselves highly correlated. This dependence further weakens the statistical evidence and suggests that many of the reported *t*-statistics may not be statistically significant. The estimate of (b) reported in Fig. 4 of \cite{Schmidhuber} is not inconsistent with the model of Grebenkov, which does not require a mean-reversion term toward a fundamental value. Rather, it simply indicates that momentum signals computed over longer lookback horizons have lower predictive power}, even though they also reported empirical results in Table 2 by measuring the Sharpe ratio of the portfolio as a function of the indicator parameter. The shape of this relationship is consistent with our results and with the \cite{Grebenkov} model.}
 
{\clh We can also note that $\beta_0=0.12$ is higher than $0.08$, the parameter measured for the Dow Jones over the past 100 years by \cite{Grebenkov}. We would have expected $\beta_0$ to be closer to $0.4$ instead of $0.12$ when applying the scaling factor from $N=1$ to $N=73$, but the markets seem to have been much more inefficient 100 years ago than over the past 20 years. The decline in $\beta_0$ over time is illustrated in Section \ref{decline} of the Appendix.  The difference matters: with Eq. \eqref{eq:optimal}, our parameters ($\lambda=1/180$, $\beta_0=0.12$) give a maximum annualized Sharpe ratio of about $1.25$, against about $0.56$ for the Dow Jones parameters of \cite{Grebenkov} ($\lambda=0.011$, $\beta_0=0.08$), i.e. a Sharpe ratio about $2.2$ times higher when applying the strategy to a universe of 73 underlying assets instead of the Dow Jones. As for $\lambda$, our analysis is that the stock market index was measured faster due to noise but also that market behavior should have changed in the last 100 years, and we believe that our fit is more robust than a simple variogram. This further confirms the importance of measuring the implied autocorrelation parameter based on a strategy invested in a large universe but a more recent period, as it leads to more accurate and aggregated estimates.}


We can see based on the Tab.\ref{tab:tab1}, in appendix \ref{results}, that empirical Sharpe is $1.24$ for $\ARP\left(120\right)$ with one time scale and $1.18$ for $\MACD\left(20,120,0.4\times 400\right)$ with 3 time scales.  $\MACD$ could not be justified as additional time scales do not bring significant improvement. Also Sharpe ratio is not so  sensitive to the parameters around the optimal as expected: Sharpe is $1.25$ for  $\ARP\left(100\right)$ and $1.21$ for $\ARP\left(150\right)$. 
 
 

 \begin{figure}[H]
 	\centering
 	\includegraphics[width=0.8\textwidth]{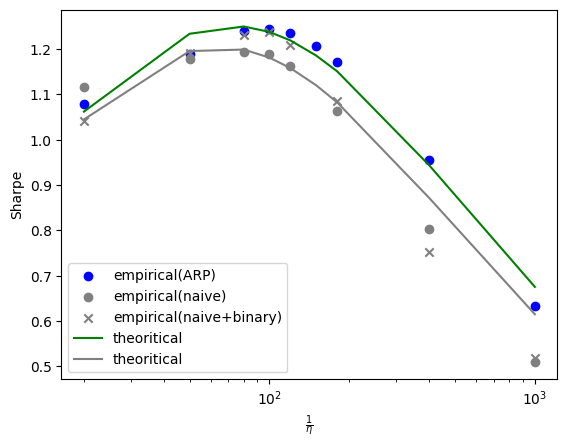}
 	\caption{Empirical Sharpe ratio based on the whole period 1990-2023 for trend-following strategies using $\ARP$ (Eq. \eqref{eq:arp} ) or the naive '1/N' approach applied to $\EMA$ as signal (both with the linear trading rule and the binary trading rules)  v.s. $\eta$ the parameter of the $\EMA$ (Eq. \eqref{eq:ema} ) and the theoretical equation Eq. \eqref{eq:optimal}. The theoretical model is fitted with parameters $\beta_0=0.12$ and $\lambda=1/180$ so that $\eta_{\text{opt}}=\frac{1}{72}$ . Empirical measurements for the $\ARP$ are displayed in the Tab.\ref{tab:tab1}. {\clh The green line is the theoretical curve fitted on the $\ARP$ measurements and the grey line the theoretical curve fitted on the naive '1/N' measurements.}} 
 	\label{fig:empirical_theory}
 \end{figure}

 \subsection {Scaling factor empirical validation }
 \label{scaling}

  {We can see that the empirical measurements in Fig.\ref{fig:empirical_theory} appear less noisy than those reported by \cite{Ferreira18} for the Dow Jones, where a complex ARMA process was required instead of a simple AR(1). It is clear that aggregating data by projecting onto our optimal trend-following portfolio, rather than analyzing each market individually, is highly effective in reducing noise that could otherwise create misleading evidence of mean reversion at longer time scales, as observed in \cite{Moskowitz}.
  	
  	We applied a simple empirical protocol to illustrate why analyzing the Sharpe ratio as a function of model parameters at the portfolio level is far more efficient than doing so market by market, as in \cite{Ferreira18} (see their Fig.6), which was affected by considerable noise. Another advantage of this protocol is that it provides an empirical test of the
  	rescaling of $\beta_0$ with the size of the universe proposed in
  	Section \ref{theory}, Eq. \eqref{eq:beta0N}. We also report, for comparison, the
  	simpler rescaling of the Sharpe ratio itself by
  	$\sqrt{\frac{N}{1+\left(N-1\right)\bm{\rho}^2}}$, which is what a portfolio of
  	$N$ strategies with an average squared correlation $\bm{\rho}^2$ would give.

  	We measured the empirical Sharpe ratios of the $\ARP$ using an $\EMA$ parameter value of $\frac{1}{120}$, which lies between the theoretical optimum and the value of $\lambda$ for which the Sharpe ratio should be very close to $\beta_0 \sqrt{\frac{255}{2}}$, as discussed in Section \ref{theory}. We computed the empirical Sharpe ratios across several randomly generated universes containing 1, 3, 6, 9, 15, 20, or 27 underlying assets, each selected randomly from the original universe of 73 assets. For each universe size, we conducted 20 independent random trials.

  	Fig.\ref{fig:empirical_theory_scaling} %
{\clf shows another good fit},  when rescaling directly the Sharpe ratio, although the fitted value of $\bm{\rho}^2$ ($0.024\pm 0.012$) appears slightly lower compared with the empirical average of squared correlations ($0.056$) based on the weekly returns. We can also observe that the fit obtained with the simpler square-root formula is less accurate and would predict larger Sharpe ratios for larger universes.

When fitting Eq. \eqref{eq:optimal} combined with Eq. \eqref{eq:beta0N}, with
$\lambda=\frac{1}{180}$ and $\eta=\frac{1}{120}$, we obtain
$\bm{\rho}^{2}=0.09$ and $\beta=0.004$. This estimate of $\bm{\rho}^{2}$ is closer to the empirical
average of squared correlations ($0.056$) than the one obtained with the simple
square-root formula, although it now lies above it.

 The fit of the Sharpe rescaling leads to a Sharpe ratio of $1.28$ for the whole
 universe of size $N=73$, $1.40$ for $N=140$ and $1.60$ for an infinite universe.
 The fit of the $\beta_0$ rescaling, using Eq. \eqref{eq:optimal} combined with
 Eq. \eqref{eq:beta0N}, gives a different picture: $1.11$ for $N=73$, $1.15$ for
 $N=140$ and only $1.21$ for an infinite universe. The explanation is that
 Eq. \eqref{eq:optimal} holds only for small $\beta_0$, and that the assumption
 that trends are captured by a single common factor further limits the effect of
 diversification. 
  	
  	The same Fig.\ref{fig:empirical_theory_scaling} also shows that the Sharpe ratio increases with the size of the universe, while its standard deviation, theoretically expected to remain stable around 0.2, declines only slightly due to the larger correlations between  randomly generated universes when they are larger. We can see that the mean Sharpe ratio is only about 0.2 when $N=1$, the one dimensional case, which is of the same order of magnitude as the noise level. This highlights how diversification and data aggregation—through projection onto the $\ARP$—help reduce the noise-to-signal ratio of the measurements when $N$ is large.  We can also see that the exact standard deviation at $N=1$ of 0.2087 is not significantly larger than the standard deviation of 0.175, just explained by the noise of Sharpe ratio %
{\clh using 32.58 years of simulated returns (approximately the first year of data is used to initialize the correlation matrix and the trend indicators and is not included in the measured performance)}, that we would also obtain with an infinite number of random trials instead of 20, assuming that the autocorrelation coefficients were universal across asset classes. This last observation, which does not conflict with the assumption of universal autocorrelation parameters across different sectors, justifies our choice to aggregate data across asset classes.}
 
 \begin{figure}[H]
 	\centering
 	\includegraphics[width=0.8\textwidth]{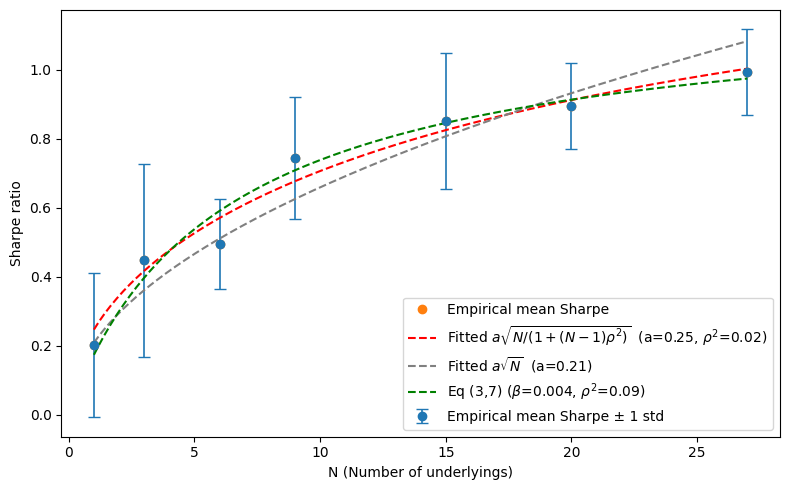}
 	\caption{Empirical Sharpe ratio based on the whole period 1990-2023 for trend-following strategies using $\ARP$ (Eq. \eqref{eq:arp} ) applied to $\EMA$ with $\eta=1/120$ as signal  v.s. $N$ the size of the universe and the theoretical scaling factor. The fitted $\bm{\rho}^2$  at 0.09 is just higher than the empirical average of squared correlations. }
 	\label{fig:empirical_theory_scaling}
 \end{figure}

 \subsection{Replication of a simple $\EMA$  by a mixture of {B}ollinger {B}ands $\BB$ }
The Bollinger Bands indicator, $\BB$, {\clr which is most likely the most popular technical indicator in the financial industry,} applies a double {\clh Heaviside} function to an $\SMA$ with a width $\delta$. This indicator is nonlinear and therefore complex, with sensitivities to past returns that depend on the price path. However, the usual approach is to use a signal composed of a mixture of many Bollinger Bands indicators with different parameters, making the signal less path-dependent and increasingly robust. Here, we aim to demonstrate that a simple $\EMA$ signal can be replicated through a complex mixture of a large number of  $\SMA$, which can, in turn, be decomposed into a large number of $\BB$ Bollinger Bands (Eq. \eqref{eq:bb} ). This explains why it is common, as shown in Fig.\ref{fig:BB} {\clh when replicating the optimal} $\EMA$ with 72 days, to display indicator weights following a bell-shaped curve centered around 120 days, while explaining to investors that the signal contains both short-term and long-term indicators.
\begin{small}
	\begin{eqnarray}
		\EMA_{t}\left(\eta\right) =\frac{r_t + \left(1-\eta\right) r_{t-1}+\left(1-\eta\right)^2 r_{t-2}+....+\left(1-\eta\right)^n r_{t-n}+....}{1 + \left(1-\eta\right) +\left(1-\eta\right)^2 +....+\left(1-\eta\right)^n +...} \\
		\EMA_{t}\left(\eta\right) =\frac{...+\left[\left(1-\eta\right)^{n-1}\eta\right] \left(n-1\right) \SMA_{t}\left(n-1\right) +...+\left[\left(1-\eta\right)\eta\right] \SMA_{t}\left(1\right) }{1 + \left(1-\eta\right) +\left(1-\eta\right)^2 +....+\left(1-\eta\right)^n +...} \label{eq:bb0}\\
		\SMA_{t}\left(n\right)=\int_{0}^{\infty}  \BB_{t}\left(n,\delta\right) \, d\delta
		\label{eq:bb}		
	\end{eqnarray}
\end{small}

\begin{figure}[H]
	\centering
	\includegraphics[width=0.8\textwidth]{ 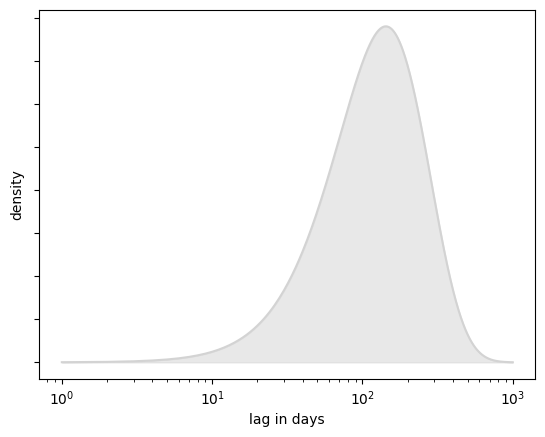}
	\caption{Weights in the mixture of $\BB$  replicating the $\EMA\left(\frac{1}{72}\right)$ Vs size of the window of elementary $\BB$ (size or lag in logarithm) . The weights are derived from Eq. \eqref{eq:bb0} and pick the width parameter of $\BB$ {\clh uniformly}.  }
	\label{fig:BB}
\end{figure}

 \subsection{Non-linear findings which challenge the $\EMA$ }
\label{nonlinearsection} 

 {\cite{Schmidhuber} provides empirical evidence across different time scales suggesting that when trends are extreme, the optimal indicator becomes negatively dependent on the trend itself. As a consequence, strong trends appear to be relatively rare. Specifically, for a normalized trend variable $\phi$, they find that the optimal signal depends on a nonlinear transformation of the form $\phi - c\phi^3$, with $c = 0.33$. They regress daily normalized returns on both the normalized trend and its cubic term as explanatory variables, across various underlyings and asset classes, including commodities.}

 {While they use a different trend indicator than $\phi$ derived in Eq. \eqref{eq:ema}, our version of the renormalized $\EMA$, their results challenge the assumption of a purely gaussian diffusive return process from Eq. \eqref{eq:autoco}, which underpins the theoretical optimality of $\EMA$-based strategies ({i.e., we would obtain $\beta$ from Eq. \eqref{eq:autoco} not as a constant, but as a path-dependent quantity that depends on $\sum_{k=1}^{t-1} (1-\lambda)^{t-1-k} \xi_{i,k}$}). Interestingly, when applying the Agnostic Risk Parity ($\ARP$) strategy using $\phi - c\phi^3$ instead of $\phi$ as the trend signal, %
{\clh we observe no significant improvement in the Sharpe ratio}\footnote{{\clrr ARP with $\eta=1/120$ will generate a Sharpe ratio of $1.24$ for $c=0$, $1.24$ for $c=0.3$ and $1.14$ for $c=0.4$. ARP will generate a Sharpe ratio of $1.24$ for $\gamma=0$, $1.28$ for $\gamma=1$ (Taylor expansion of $\tanh$ would give $\gamma=1$ the closer to $c=0.33$) and $1.29$ for $\gamma=1.3$  where $\gamma$ is the non linear parameter inside the $\tanh$ as in \cite{Kurth2}. %
{\clh These differences are small compared with the statistical error of about $0.18$ on each Sharpe ratio and are not significant.}}}. This discrepancy might be due to the lack of statistical significance of the coefficient $c$ (the sensitivity to the cubic term), which is highly sensitive to a few extreme trend observations. {\clrr However, this explanation is unlikely, as similar findings were originally reported by \cite{Ferson,Lamperiere} and subsequently confirmed by \cite{Bouchaud2017,Kurth2,Moskowitz2}\footnote{It is conceivable that the estimates of the nonlinearity parameter $\gamma$ (initially estimated from 2 to 7 depending on the futures as documented in the Tab. 1 of \cite{Kurth2}) and the mean-reversion coefficient f reported by \cite{Kurth2} would be reduced if the momentum dynamics were extended to include an idiosyncratic noise term, i.e., replacing $dM_t=-\alpha M_t dt+\alpha\left(dP_t-g_t dt\right)$ in $dP_t=f\left(V_t-P_t\right)dt+\beta \tanh\left(\gamma M_t\right) dt+g_t dt+\sigma_N dW_t^N$ and $dV_t=g_tdt+\sigma_V dW_t^V$   by   $dM_t=-\alpha M_t dt+h\alpha\left(dP_t-g_t dt\right)+\sigma_MdW_t^M$, thereby bringing the specification closer to the latent-state formulation of \cite{Grebenkov15}. Such an extension would allow part of the variability currently attributed to the deterministic response to price changes to be absorbed by an independent stochastic component.}. It may instead result from the loss of the rotational invariance property of the $\ARP$ : applying the nonlinear transformation $\phi - c\phi^3$
 		either to the raw returns or to their projections onto the eigenvectors of the covariance matrix no longer yields equivalent results because of the nonlinearity of the transformation. Under this interpretation, the nonlinear transformations proposed by \cite{Ferson,Moskowitz2} would simply modify the portfolio weights, as the $\ARP$ already does (large trends would be reduced by ARP because they are most of the time driven by a common factor), thereby improving the Sharpe ratio. Applying both corrections simultaneously, however, would be redundant and therefore unlikely to provide additional benefits. Another possible explanation is the difference in the normalization used to measure the relationship between future returns and past trends. We normalize in the $\EMA$ signal using a short-term volatility estimate to correct for the heteroskedasticity of returns and to better approximate the Gaussian case, whereas \cite{Schmidhuber22,Bouchaud2017} normalize returns using the standard deviation computed over the entire sample. That last explanation is strongly supported by \cite{Lamperiere}, who also take heteroskedasticity into account in their normalization and report much weaker nonlinearities in their Fig. 5. Finally, the discrepancy may arise from measurement errors in the estimation of the unobserved trend, induced by the term $\sum_{k=1}^{t-1} (1-\lambda)^{t-1-k} \xi_{i,k}$ in Eq. (\ref{eq:autoco}), which may distort the apparent relationship when returns are non-Gaussian.}

 {\clr One interesting theoretical justification for this introduced nonlinearity is to keep financial markets close to efficiency despite herding behavior. The classical and earlier approach was to complexify Eq. \eqref{eq:autoco} through introducing} mean-reverting behavior from fundamentalist traders and assuming an attraction toward a fair or efficient price though such a price is unobservable ({\clrr \cite{Majewski} implement the two effects to evaluate their impact which generates  non-monotonic relation between past trends  and future returns}). It can be approximated by a moving average, which effectively induces negative autocorrelation, especially over longer time scales. \cite{Zakamuli} explored this idea, though obtaining statistically significant results remains difficult. By contrast, \cite{Schmidhuber}'s approach involves making autoregressive coefficients trend-dependent, enabling a dynamic transition between trend-following governed by trend-followers and mean-reversion regimes governed by fundamentalist traders. Another way to model the influence of fundamental investors could be simply to assume their impact manifests in changes to the trend, i.e., on the noise $\xi_{i,k}$  in Eq. \eqref{eq:autoco}. It could be enough to make sure that prices are not totally disconnected from fair and efficient ones which would differ in that case from moving average of prices.

 {The findings of \cite{Schmidhuber} should not be ignored. A deeper investigation is, however, warranted to better understand their implications for the construction of optimal trend-following portfolios, although this lies beyond the scope of the present simple work and should be the subject of a dedicated study.}

\subsection {Correlations between indicators }

Fig.\ref{fig:heat_map} in Appendix \ref{results} shows that the indicators with different parameters give very correlated strategies  {\clh which confirms} {\clb that the traditional approach of using a basket of many indicators as signal is not so appealing}. $\ARP\left(80\right)$ is correlated to $\ARP\left(150\right)$ with a coefficient of $0.96$. %
{\clf Such strong correlations are noteworthy and could be explained} by a common factor somewhere, which could be the object of an additional research. {\clh Moreover} the $\ARP\left( 120\right)$ based on a simple $\EMA$ is very close with correlation from 1 to 0.99 to  the new indicator $\MACD\left(20,120,0\times 400\right)$ or $\MACD\left(20,120,0.4\times 400\right)$  we introduced based on 3 different time scales to increase weights for very old past returns and decrease weights on very recent ones. So it makes this {\clb refinement} not justified again.

\section {Conclusion}

{\clf The data support Grebenkov's model for describing trends: its theoretical
	formula for the Sharpe ratio as a function of the $\EMA$ parameter is in good
	agreement with the empirical measurements over 1990--2023. The best fit is
	obtained with $\lambda=\frac{1}{180 \pm 17}$ and $\beta_0=0.12$. With these
	parameters, the Sharpe ratio is maximal, before trading costs, for a simple
	$\EMA$ of about $72$ business days, equivalent to a half-life of 50 business
	days (between 67 and 79 business days when $\lambda$ varies within its
	indicative confidence interval, $\beta_0$ being kept fixed). This value is an
	average over the sample, and the current optimal parameter is more uncertain:
	since $\beta_0$ declines over time, it could be as slow as 120 business days
	(a half-life of 83 business days) if the current value of $\beta_0$ were half
	the value estimated on our sample. The half-life of 60 to 70 business days that
	best replicates the returns of trend-following CTA industry benchmarks,
	according to \cite{quantica}, lies within this range. The Sharpe ratio is in any
	case not very sensitive to the parameter around its optimum.
	
	It is quite surprising that a single $\EMA$ is sufficient to capture trends, as
	one would expect different time scales for different types of investors or
	long-term mean reversion at longer time scales, as documented by
	\cite{Moskowitz}. There are likely much shorter time scales, on the order of a
	few days, as documented, for example, by \cite{Kurth,Zhao,Rossi2026}, but they have no
	significant impact on a medium-frequency strategy. The conclusion is that, on
	these data, a complex mixture of sophisticated indicators does not improve on a
	single $\EMA$.}

\section {Declaration of funding}

No funding was received



\appendix

\section{Review of Literature of Solving Markowitz Optimization with Uncertainty of the Parameters}

The sensitivity of the classical \cite{Markowitz52} mean–variance optimization framework to estimation errors in expected returns and covariances has long been recognized. When means and covariances are replaced by their empirical estimates, the resulting optimal portfolios often exhibit extreme positions and poor out-of-sample performance. This issue has motivated a wide array of research aiming to account for parameter uncertainty in portfolio selection.

\subsection*{Bayesian and Empirical–Bayes Approaches}

The Bayesian framework addresses parameter uncertainty by placing prior distributions on the mean vector and, sometimes, the covariance matrix. \cite{Jorion86} introduced a Bayes–Stein shrinkage estimator in which sample means are shrunk toward a common value, yielding more stable estimates of expected returns. Later, \cite{BlackLitterman92} proposed a practical Bayesian model combining equilibrium-implied returns with subjective investor views, producing posterior expected returns that are robust to estimation noise. \cite{KanZhou07} formalized the incorporation of estimation risk directly into the optimization problem, deriving analytical adjustments to account for parameter uncertainty. However, these Bayesian and empirical–Bayes methods only use uncertainty information to update or regularize the input parameters of the Markowitz optimization; the structure of the optimal solution itself, $w^\star \Sigma \propto   C^{-1}\mu/ \Sigma$, remains unchanged.

\subsection*{Covariance Shrinkage and Regularization Approaches}

A complementary literature focuses on improving the estimation of the covariance matrix. \cite{LedoitWolf03,LedoitWolf04} developed optimal shrinkage estimators combining the sample covariance matrix with structured targets, leading to more stable and better-conditioned inputs for portfolio optimization. \cite{JagannathanMa03} showed that imposing simple portfolio constraints (such as nonnegativity) implicitly regularizes the covariance matrix, mitigating estimation error effects. These approaches improve empirical robustness but, again, modify only the inputs of the optimization, not the analytical form of the optimal portfolio weights. Related structured-ensemble, predict-then-optimize methods steer diversification \emph{ex ante} by tuning diversity in prediction and selection, using the ensemble’s combiner (equal-weight under MSE) as the target, while leaving the Markowitz closed form intact \cite{RODRIGUEZDOMINGUEZ2025128633}.

\subsection*{Robust Optimization and Worst-Case Formulations}

Another important strand of research uses robust optimization to address parameter uncertainty. \cite{Goldfarb03} and \cite{ElGhaoui03} recast the mean–variance problem as a convex program that optimizes for the worst-case scenario within an uncertainty set for the parameters. Robust optimization offers clear protection against misspecification of $\mu$ and $C$, but the resulting portfolios correspond to a modified optimization problem rather than a redefinition of the Markowitz formula itself. The uncertainty is modeled exogenously, not endogenously reflected in the structure of the optimal solution. 

The robust mean variance objective function is then typically
\[
\min_{w} \;
w^{\top} (C ) w
- \lambda\, \hat{\mu}^{\top} w
+ \lambda \rho\, \| Q^{1/2} w \|_2
\]
It can be decomposed into three distinct components, each corresponding to a different source of risk or model uncertainty:

\begin{enumerate}
	\item \textbf{Covariance risk term:} 
	The quadratic form $w^{\top} (C ) w$ represents the classical portfolio variance.

	\item \textbf{Expected return term:}
	The linear component $-\lambda\, \hat{\mu}^{\top} w$ corresponds to the standard mean--variance trade-off.
	The coefficient $\lambda$ controls the investor’s risk aversion and governs the balance between expected return and variance minimization.
	
	\item \textbf{Uncertainty penalization term:}
	The additional penalty $\lambda \rho\, \| Q^{1/2} w \|_2$ arises directly from the ellipsoidal uncertainty set on the mean vector. So $ \lambda\, \hat{\mu}^{\top} w
	- \lambda \rho\, \| Q^{1/2} w \|_2$ is the worst-case scenario when the $\mu$ are set in a conservative way.
	Intuitively, it reduces exposure to directions in which the mean estimates are most uncertain (as encoded by the covariance matrix $Q$ of estimation errors).
	The scalar $\rho$ acts as a robustness parameter: larger values of $\rho$ correspond to more conservative portfolios that perform well even under unfavorable perturbations of $\hat{\mu}$.
\end{enumerate}

In summary, the first term controls sensitivity to risk, the second captures expected performance under the nominal model, and the third explicitly regularizes against uncertainty in the mean estimates.
Together, these terms yield a portfolio that trades off performance and robustness in a principled convex optimization framework.{ \clr \cite{Segonne} also uses $Q$ in a different way and introduces an anisotropy penalty into the mean-variance objective function to improve robustness. The term $\lambda \rho \,\| Q^{1/2} w \|_2$ is transformed into an anisotropy component through his Eq. (48), which also helps interpret the solution of \cite{Benichou16} as a robust one}. To our knowledge, this "Robust Optimization and Worst-Case Formulations"  is the only approach in the literature that explicitly accounts for the covariance matrix  $Q$  of estimation errors, although the results differ from those of \cite{Benichou16} and \cite{Valeyre}.

\subsection*{Towards Uncertainty-Dependent Optimal Formulas}

While the above approaches (except for the robust approach, which, however, neither relates directly to the optimal solution knowing uncertainties) treat uncertainty through priors, constraints, or uncertainty sets, they do not alter the analytical dependence of the optimal portfolio on $\mu$ and $C$. 

{\clr \cite{Brandt} originally introduced the idea of optimising a portfolio under the assumption that portfolio weights are proportional to time-dependent signals $x_t$ (not only trend-following signals). In their framework, this is equivalent to determining an optimal matrix $A$ such that, in their notation, $w_t \propto A x_t$ yields the optimal "dynamic" portfolio. They derived a very complex expression for $A$, although it nonetheless resembles the structure of the Markowitz solution.

\cite{Grebenkov15} also optimises the portfolio by determining a matrix $A$ under the same proportionality assumption, but applies it specifically to the Grebenkov diffusive model and to exponential moving-average signals. Their results were derived independently from those of \cite{Brandt}. Under suitable approximations, \cite{Valeyre} showed that this matrix can be written as $A=C^{-1} C_{\xi} C^{-1}$, where $C_{\xi}$ is the covariance matrix of the estimation uncertainties $\xi$ of Eq. \eqref{eq:autoco}. By assuming a random diffusive model for $C_{\xi}$, \cite{Valeyre} obtained the $\ARP$ portfolio, for which $A=C^{-1} C_{\xi} C^{-1}$ reduces to $A=C^{-1/2}$.

Subsequently, \cite{Firoozye23, Kelly22} proposed simplified approximations from results of \cite{Brandt} leading to an easy-to-use "generalised" Markowitz formula that properly incorporates $Q$ through their Eq. (6), namely $A=C^{-1} C_{rs} Q^{-1}$ (with our notations), instead of the $A=C^{-1}$ appearing in the classical Markowitz formula. Here, $C_{rs}$ denotes the "cross" variance matrix between returns and signals.

The expressions $A=C^{-1} C_{rs} Q^{-1}$ and $A=C^{-1} C_{\xi} C^{-1}$ become equivalent—or at least very close—under appropriate approximations when the signals are generated both by the Grebenkov diffusive process and by exponential moving averages. One may therefore say that \cite{Firoozye23, Kelly22} generalised the Markowitz formula, although this contribution has remained relatively unnoticed. Their work also generalises \cite{Grebenkov15}, as it is not restricted to trend-following signals nor to the Grebenkov diffusion framework.

All these unifying relationships between the different approaches of \cite{Grebenkov15, Kelly22, Firoozye23, Valeyre, Benichou16} are clearly explained in \cite{Segonne}. Ultimately, it appears critical to model $C_{\xi}$ (or both $C_{rs}$ and $Q$) appropriately, since these matrices are largely unknown and cannot be measured with precision without causing overfitting to the portfolio. Modelling them through a random diffusive process naturally leads to the solution of \cite{Benichou16}, ensuring that the resulting optimal portfolio is both robust and genuinely out-of-sample optimal. Empirical backtests, in the context of trend-following, confirm that this approach delivers strong out-of-sample performance.

\cite{Benichou16} and \cite{Valeyre} therefore propose frameworks in which the optimal portfolio explicitly incorporates estimation uncertainty into its structure. These works derive solutions that differ fundamentally from the classical Markowitz prescription. In particular, under specific random matrix assumptions on the structure of the estimation-error correlation matrix, the optimal weights become proportional to $A=C^{-1/2}$ rather than $A=C^{-1}$, thus providing a substantial conceptual revision of the Markowitz approach. This line of research, %
{\clf in which random matrix theory plays a central role}, remains rare in the literature but represents a promising direction for embedding measurement uncertainty at the core of portfolio theory.}

\section{{\clh Comparison of Grebenkov's Formula with Those of the Literature}} 
\label{comparaisonwithortherformulas}

{\clf This appendix relates the Sharpe formula of \cite{Grebenkov} to the other closed formulas of the literature. Two routes lead to the Sharpe ratio of a trend-following rule. \cite{Acar} starts from the joint Gaussian distribution of the forecast $F$ and of the future return $X$: the Sharpe ratio of the binary rule is then a function of their correlation $\rho_{xf}$, whatever the process that generates the returns, and this correlation remains to be computed for each process and each indicator. \cite{Grebenkov} start instead from a given process and a given indicator (the $\EMA$) and compute directly the mean and the variance of the P\&L of the linear rule, written as a quadratic form of Gaussian variables. Section \ref{acarsection} recalls the formula of \cite{Acar}. Section \ref{workofjulien} computes $\rho_{xf}$ for the process of \cite{Grebenkov} and the $\EMA$, which gives the Sharpe ratio of the binary rule in closed form. Section \ref{proofhold} follows the route of \cite{Acar} for the linear rule \citep{Firoozye} and shows that the result is identical to the exact expression of \cite{Grebenkov}, Eq. \eqref{eq:optimal} being its approximation for small parameters. It also links the two rules: their Sharpe ratios are increasing functions of the same correlation, so that they share the same optimal $\EMA$, which is the stationary Kalman filter of the model. Finally, we recall the related formulas of \cite{Ferreira18} (linear rule, simple moving average) and \cite{Zakamuli22} (binary rule, momentum).}

\subsection{	\cite{Acar}} 
\label{acarsection} 

			\cite{Acar} uses a binary "Long/Short" strategy (Eq. \eqref{eq:H} ) in Section 8.5.1, "Optimal long/short fund on the index itself $X$" using his notations, where $F$ is the forecast used to predict the index $X$ and $H$ is the binary strategy. \cite{Acar} presents his equation (Eq. \eqref{eq:SharpeAcar} ) in Section 8.1, based on \cite{Acar98}.

\begin{equation}
	H =
	\begin{cases}
		X & \text{if } F > 0,\\
		-X & \text{if } F < 0.
	\end{cases}
	\label{eq:H}
\end{equation}

\begin{eqnarray}
	E(H_{\max}) &=& \mu_x \left( 1 - 2 \, \phi\Big(-\frac{\mu_x}{\sigma_x \rho_{xf}}\Big) \right)
	+ \sigma_x \sqrt{\frac{2}{\pi}} \rho_{xf} \exp\Big(-\frac{1}{2} \frac{\mu_x^2}{\sigma_x^2 \rho_{xf}^2} \Big), \\
	SR(H_{\max}) &=& \frac{E(H_{\max})}{\sqrt{\mu_x^2 + \sigma_x^2 - E(H_{\max})^2}}.
	\label{eq:SharpeAcar}
\end{eqnarray}

%
{\clf In Eq. \eqref{eq:SharpeAcar}, $\rho_{xf}$, the correlation between the returns $X$ and the forecast $F$, remains to be expressed as a function of the parameters of the return process and of the indicator. We do this in Sect. \ref{workofjulien} for the process of \cite{Grebenkov} with the $\EMA$ as indicator, which leads to Eq. \eqref{monstruousoone} and Eq. \eqref{monstruoustwo}.}

\subsection{ Closed formula derived from \cite{Acar} and \cite{Grebenkov} process}
\label{workofjulien}


Throughout this appendix we assume $0<\eta<1$ and $0<\lambda<1$, and we write $p=1-\eta$ and $q=1-\lambda$, so that $0<p,q<1$ and all the geometric series below converge.

\paragraph*{Acar's formula} Let $F$ be a forecast for the price of an instrument and $X$ be its future returns. Let $a$ and $b$ be two real numbers. Let us assume that the vector $(X,F)$ follows a bivariate normal distribution and let $H$ be the returns generated by the binary forecasting rule that buys $a$ units of the instrument if $F>0$ and buys $b$ units if $F<0$. \newline
\cite{Acar} gives closed formulas for the first two moments of $H$ (hence for its Sharpe ratio) as functions of the means, the variances and the correlation of $X$ and $F$ (resp. $\mu_x$, $\mu_f$, $\sigma_x^2$, $\sigma_f^2$ and $\rho_{xf}$):
\begin{align}
	\E(H) &= \mu_x \left[ a\Phi\!\left( \frac{\mu_f}{\sigma_f} \right)
	+ b\Phi\!\left( \frac{-\mu_f}{\sigma_f} \right) \right]
	+ \sigma_x \frac{(a - b)}{\sqrt{2\pi}} \rho_{xf}
	\exp\!\left( -0.5 \frac{\mu_f^2}{\sigma_f^2} \right) \\[1em]
	\E(H^2) &= \mu_x^2 \left[ a^2\Phi\!\left( \frac{\mu_f}{\sigma_f} \right)
	+ b^2\Phi\!\left( \frac{-\mu_f}{\sigma_f} \right) \right]
	+ 2\mu_x\sigma_x \frac{(a^2 - b^2)}{\sqrt{2\pi}} \rho_{xf}
	\exp\!\left( -0.5 \frac{\mu_f^2}{\sigma_f^2} \right) \nonumber \\[0.5em]
	&\quad + \sigma_x^2 \Bigg[
	a^2 \left( \frac{\rho_{xf}^2}{\sqrt{2\pi}}
	\left( -\frac{\mu_f}{\sigma_f}\right) \exp\!\left(-0.5\frac{\mu_f^2}{\sigma_f^2}\right)
	+ \Phi\!\left(\frac{\mu_f}{\sigma_f}\right)\right) \nonumber \\[0.5em]
	&\quad\quad + b^2 \left( \frac{\rho_{xf}^2}{\sqrt{2\pi}}
	\left( \frac{\mu_f}{\sigma_f}\right) \exp\!\left(-0.5\frac{\mu_f^2}{\sigma_f^2}\right)
	+ \Phi\!\left(-\frac{\mu_f}{\sigma_f}\right)\right)
	\Bigg]
\end{align}

In particular, when $a=1$ and $b=-1$, we have $H^2=X^2$ and the formulas read:
\begin{equation}\label{Acar_binary}
	\E(H) = \mu_x \left[ 1 - 2\Phi\left( -\frac{\mu_f}{\sigma_f} \right) \right]
	+ \sigma_x \sqrt{\frac{2}{\pi}}\, \rho_{xf} \exp\left( -0.5 \frac{\mu_f^2}{\sigma_f^2} \right)
\end{equation}
\begin{equation}\label{Acar_binary_SR}
	SR(H) = \frac{\E(H)}{\sqrt{\mu_x^2 + \sigma_x^2 - \E(H)^2}}
\end{equation}
(The expression quoted above for $H_{\max}$ in the subsection on \cite{Acar} is recovered when $\mu_f/\sigma_f = \mu_x/(\sigma_x\rho_{xf})$.) In the centered case $\mu_x=\mu_f=0$, these reduce to
\begin{equation}\label{Acar_binary_centered}
	\E(H) = \sqrt{\frac{2}{\pi}}\,\sigma_x\,\rho_{xf},
	\qquad
	SR(H) = \frac{\sqrt{2/\pi}\;\rho_{xf}}{\sqrt{1-\frac{2}{\pi}\rho_{xf}^2}} .
\end{equation}

\paragraph*{Grebenkov's dynamic}
In \cite{Grebenkov}, the authors assume the following dynamics for the returns of the underlying price (Eq. (3), p.5 of \cite{Grebenkov}):
\begin{equation}
	r_t = \varepsilon_t + \beta\sum_{k=1}^{t-1} (1-\lambda)^{t-1-k} \xi_k,
\end{equation}
where the $\varepsilon_k$'s and the $\xi_k$'s are i.i.d. following $\mathcal{N}(0,1)$.

In matrix form, we can write the vector of the returns $\mathbf{r}$ as:
\begin{equation}
	\mathbf{r} = \varepsilon + \beta \mathbf{E}_{1-\lambda}\xi, 
\end{equation}
where $\mathbf{E}_q$ is the $T\times T$ matrix whose coefficients $(\mathbf{E}_q)_{j,k}$ are $0$ if $j\leq k$ and $q^{j-k-1}$ otherwise. \newline
We directly see from the matrix expression that $\mathbf{r}$ is a centered gaussian vector, and that its covariance matrix $\mathbf{C}$ satisfies (Eq. (6), p.6 of \cite{Grebenkov}):
\begin{equation} \label{cov_matric_formula}
	\mathbf{C} = \mathbf{I}_T + \beta^2 \mathbf{E}_{1-\lambda}\mathbf{E}_{1-\lambda}^T. 
\end{equation}
In particular, for two indices $j,k$, we have:
\begin{equation}\label{cov_entries}
	\mathbf{C}_{j,k} = \delta_{j, k} + \beta_0^2\left[ (1-\lambda)^{|j-k|} - (1-\lambda)^{j+k-2}\right],
	\qquad
	\beta_0 := \frac{\beta}{\sqrt{\lambda(2-\lambda)}} ,
\end{equation}
where, as in \cite{Grebenkov}, the rescaled parameter $\beta_0$ makes the asymptotic variance of the returns independent of the timescale $\lambda$. In the stationary regime ($j,k\rightarrow\infty$ with $|j-k|$ fixed) the second term in the bracket vanishes and
\begin{equation}\label{cov_stationary}
	\langle r_j, r_k\rangle = \delta_{j,k} + \beta_0^2\, q^{|j-k|},
	\qquad
	\sigma_{x,\infty}^2 := \langle r_\infty, r_\infty\rangle = 1 + \beta_0^2 .
\end{equation}

\paragraph*{Momentum signal}
Given this dynamic, Grebenkov defines a momentum signal as an EMA of the returns as follows (Eq. (13), p.8):
\begin{equation}
	s_t = \gamma \sum_{k=1}^{t-1} (1-\eta)^{t-1-k}r_k, \qquad \gamma>0 .
\end{equation}

\subsubsection{From Acar's formula to Grebenkov's}
\paragraph{Covariance computation}
Using Eq. \eqref{cov_entries} and summing the two geometric series, we have for $p\neq q$:
\begin{align}
	\langle s_t, r_t\rangle 
	&= \gamma \sum_{k=1}^{t-1} p^{t-1-k}\, \mathbf{C}_{k,t} 
	= \gamma \beta_0^2 \left( \sum_{k=1}^{t-1} p^{t-1-k}q^{t-k} - \sum_{k=1}^{t-1} p^{t-1-k}q^{t+k-2} \right) \nonumber \\
	&= \gamma \beta_0^2 \left( q\,\frac{1-(pq)^{t-1}}{1-pq} \;-\; q^{t-1}\,\frac{p^{t-1}-q^{t-1}}{p-q} \right)
\end{align}
(the case $p=q$ follows by continuity). Letting $t\rightarrow\infty$, the border terms vanish and we obtain:
\begin{equation}\label{cov_sr_stationary}
	\langle s_\infty, r_\infty\rangle = \gamma\beta_0^2\,\frac{q}{1-pq}
	= \gamma\beta_0^2\frac{1-\lambda}{1-(1-\eta)(1-\lambda)}.
\end{equation}

\paragraph{Momentum's variance computation}
We compute the variance of the signal directly in the stationary regime, which is the only one needed below (the finite-$t$ expression is obtained with the same geometric sums applied to Eq. \eqref{cov_entries}, and differs only by border terms that vanish as $t\rightarrow\infty$).
\begin{proposition}\label{prop:var_signal}
	In the stationary regime,
	\begin{equation}\label{var_s_stationary}
		\sigma_{f,\infty}^2:=\langle s_\infty, s_\infty\rangle 
		= \frac{\gamma^2}{1-p^2}\left(1+\beta_0^2\,\frac{1+pq}{1-pq}\right)
		= \gamma^2\,\frac{1+\beta_0^2\,\frac{1+(1-\eta)(1-\lambda)}{1-(1-\eta)(1-\lambda)}}{1-(1-\eta)^2}.
	\end{equation}
\end{proposition}
\begin{proof}
	Writing $s_\infty = \gamma\sum_{a\geq 0} p^{a}\, r_{-a}$ (up to a shift of the time index) and using Eq. \eqref{cov_stationary}:
	\begin{align*}
		\langle s_\infty, s_\infty\rangle 
		&= \gamma^2\sum_{a,b\geq 0} p^{a+b}\left(\delta_{a,b} + \beta_0^2\, q^{|a-b|}\right) \\
		&= \gamma^2\left( \sum_{a\geq 0} p^{2a} + \beta_0^2 \sum_{a\geq 0} p^{2a}\Big(1 + 2\sum_{d\geq 1} (pq)^{d}\Big)\right) \\
		&= \frac{\gamma^2}{1-p^2}\left( 1 + \beta_0^2\Big(1 + \frac{2pq}{1-pq}\Big)\right)
		= \frac{\gamma^2}{1-p^2}\left( 1 + \beta_0^2\,\frac{1+pq}{1-pq}\right),
	\end{align*}
	where in the second line we grouped the terms $a=b$ and, by symmetry, twice the terms $b=a+d$ with $d\geq 1$.
\end{proof}

\paragraph{Signal--return correlation}
From Eqs. \eqref{cov_stationary}, \eqref{cov_sr_stationary} and \eqref{var_s_stationary}, the stationary correlation between the signal and the next return does not depend on $\gamma$ and reads:
\begin{equation}\label{rho_exact}
	\rho_{xf} = \frac{\langle s_\infty, r_\infty\rangle}{\sqrt{\langle s_\infty, s_\infty\rangle\,\langle r_\infty, r_\infty\rangle}}
	= \frac{\beta_0^2\; q\,\sqrt{1-p^2}}{\sqrt{(1+\beta_0^2)\,(1-pq)\,\big(1-pq+\beta_0^2(1+pq)\big)}} .
\end{equation}

\paragraph{Acar's formula (binary rule)}
Since the returns are centered in this model, Acar's formula for the binary rule (Eq. \eqref{Acar_binary_centered}) becomes, in the stationary regime:
\begin{align*}
	\E(H) &= \sqrt{\frac{2}{\pi}}\, \sigma_{x,\infty}\, \rho_{xf} 
	= \sqrt{\frac{2}{\pi}}\, \frac{\langle s_\infty, r_\infty\rangle}{\sqrt{\langle s_\infty, s_\infty\rangle}},
\end{align*}
that is:

\begin{multline}
	\E(H) = \sqrt{\frac{2}{\pi}}\,\beta_0^2
	\frac{(1-\lambda)\sqrt{1-(1-\eta)^2}}{
		\sqrt{\big(1-(1-\eta)(1-\lambda)\big)}}\\
	\times \frac{1}{
		\sqrt{\Big(1-(1-\eta)(1-\lambda)
			+\beta_0^2\big(1+(1-\eta)(1-\lambda)\big)\Big)}}.
	\label{monstruousoone}
\end{multline}

We then get:
\begin{equation}
	SR(H) = \frac{\E(H)}{\sqrt{1+\beta_0^2 - \E(H)^2}}
	= \frac{\sqrt{2/\pi}\;\rho_{xf}}{\sqrt{1-\frac{2}{\pi}\rho_{xf}^2}},
	\label{monstruoustwo}
\end{equation}
with $\rho_{xf}$ given by Eq. \eqref{rho_exact}.

\subsection{ Acar's formula for linear signals}
\label{proofhold}

Let us mention a useful lemma for what follows:
\begin{lemma}\label{lemma:quadform}
	Let $Z\sim\mathcal{N}(0_N, I_N)$ be a gaussian vector and let $M\in\mathcal{S}_N(\R)$ be an $N\times N$ symmetric matrix (not necessarily positive-semidefinite). Then one has:
	\begin{equation}
		\E\left[ Z^T M Z \right] = \mathrm{tr}(M),
	\end{equation}
	\begin{equation}
		\E\left[ \left(Z^T M Z\right)^2 \right] = 2\,\mathrm{tr}(M^2) + \mathrm{tr}(M)^2,
		\qquad\text{i.e.}\qquad
		\mathrm{Var}\left[ Z^T M Z \right] = 2\,\mathrm{tr}(M^2).
	\end{equation}
\end{lemma}
\begin{proof}
	Diagonalizing $M$ in an orthonormal basis, we can write:
	\begin{equation}
		Z^T M Z = \sum_{i=1}^N \lambda_i P_i^2,
	\end{equation}
	where $P_1, \cdots, P_N$ are i.i.d. $\mathcal{N}(0,1)$ variables, and $\lambda_1, \cdots, \lambda_N$ are the (real) eigenvalues of $M$. Since $\E[P_i^2]=1$ and $\mathrm{Var}[P_i^2]=2$, we get $\E[Z^TMZ]=\sum_i\lambda_i=\mathrm{tr}(M)$ and, by independence, $\mathrm{Var}[Z^TMZ]=2\sum_i\lambda_i^2=2\,\mathrm{tr}(M^2)$.
\end{proof}

We recall that $F$ is the forecast for the price of an instrument and $X$ is its future returns\footnote{We keep the notations $\mu_x, \mu_f, \sigma_x, \sigma_f$ and $\rho_{xf}$.}. The following result was proven by \cite{Firoozye}.
\begin{proposition}\label{prop:linear}
	Let us assume that $(X,F)$ is a \textbf{centered} bivariate normal distribution and let $H$ be the returns of the (linear) strategy that trades $F$, that is:
	\begin{equation}
		H = FX.
	\end{equation}
	Then the expected value, the variance and the Sharpe ratio of $H$ are:
	\begin{equation}
		\E(H) = \rho_{xf}\sigma_x\sigma_f,
		\qquad
		\mathrm{Var}(H) = \sigma_x^2\sigma_f^2\left(1+\rho_{xf}^2\right),
	\end{equation}
	\begin{equation}\label{SR_linear}
		SR(H) = \frac{\rho_{xf}}{\sqrt{1+\rho_{xf}^2}}.
	\end{equation}
\end{proposition}
\begin{proof}
	For the computation of $\E H$, we do not need to assume that $(X,F)$ is centered.\newline
	Let us define, with obvious notations, $Z:=\left(X, F\right)^T\sim\mathcal{N}(\mu, \Sigma)$. Then the random vector:
	\begin{equation}
		\Tilde{Z} := \Sigma^{-1/2} (Z - \mu),
	\end{equation}
	follows $\mathcal{N}(0_2, I_2)$.\newline
	Then, if we define the symmetric matrix:
	\begin{equation}
		A = 
		\begin{pmatrix}
			0 & \tfrac{1}{2} \\[6pt]
			\tfrac{1}{2} & 0
		\end{pmatrix},
	\end{equation}
	we can see that:
	\begin{equation}
		H = Z^T A Z.
	\end{equation}
	Since $Z = \mu + \Sigma^{1/2} \Tilde{Z}$, we have:
	\begin{equation}
		H = \mu^T A \mu + 2\mu^T A \Sigma^{1/2} \Tilde{Z} + \Tilde{Z}^T \Sigma^{1/2}A\Sigma^{1/2}\Tilde{Z}.
	\end{equation}
	Taking the expected value of this quantity and using Lemma \ref{lemma:quadform} with $M=\Sigma^{1/2}A\Sigma^{1/2}$, we get:
	\begin{align*}
		\E H 
		&= \mu^T A \mu + 0 + \mathrm{tr}\left( \Sigma^{1/2}A\Sigma^{1/2} \right) 
		= \mu_x\mu_f + \mathrm{tr}\left( A\Sigma\right) 
		= \mu_x\mu_f + \rho_{xf}\sigma_x\sigma_f .
	\end{align*}
	
	For the variance, let us assume $\mu=0_2$, so that $H=\Tilde{Z}^T M\Tilde{Z}$. By Lemma \ref{lemma:quadform} and the cyclicity of the trace:
	\begin{align*}
		\mathrm{Var}(H) &= 2\,\mathrm{tr}\left( M^2 \right) 
		= 2\,\mathrm{tr}\left( \Sigma^{1/2}A\Sigma A\Sigma^{1/2} \right) 
		= 2\,\mathrm{tr}\left( (A\Sigma)^2 \right).
	\end{align*}
	Since
	\begin{equation*}
		A\Sigma = \frac{1}{2}
		\begin{pmatrix}
			\rho_{xf}\sigma_x\sigma_f & \sigma_f^2 \\[4pt]
			\sigma_x^2 & \rho_{xf}\sigma_x\sigma_f
		\end{pmatrix},
		\qquad
		\mathrm{tr}\left( (A\Sigma)^2 \right) = \frac{1}{2}\,\sigma_x^2\sigma_f^2\left(1+\rho_{xf}^2\right),
	\end{equation*}
	we obtain $\mathrm{Var}(H)=\sigma_x^2\sigma_f^2\left(1+\rho_{xf}^2\right)$ (equivalently, $\E[H^2]=\sigma_x^2\sigma_f^2(1+2\rho_{xf}^2)$, which is Isserlis' formula), and therefore
	\begin{equation*}
		SR(H) = \frac{\E(H)}{\sqrt{\mathrm{Var}(H)}} = \frac{\rho_{xf}}{\sqrt{1+\rho_{xf}^2}} .
	\end{equation*}
\end{proof}

\begin{corollary}\label{cor:grebenkov}
	In the stationary model of \cite{Grebenkov}, with the linear rule $H_t = s_t\, r_t$, we have:
	\begin{equation}
		\E(H) = \langle s_\infty, r_\infty\rangle,
	\end{equation}
	\begin{equation}\label{SR_linear_grebenkov}
		SR(H) = \frac{1}{\sqrt{1+\frac{\langle s_\infty, s_\infty\rangle\langle r_\infty, r_\infty\rangle}{\langle s_\infty, r_\infty\rangle^2}}}
		= \frac{\rho_{xf}}{\sqrt{1+\rho_{xf}^2}},
	\end{equation}
	where $\langle s_\infty, r_\infty\rangle$, $\langle r_\infty, r_\infty\rangle$ and $\langle s_\infty, s_\infty\rangle$ are given in closed form by Eqs. \eqref{cov_sr_stationary}, \eqref{cov_stationary} and \eqref{var_s_stationary}, and $\rho_{xf}$ by Eq. \eqref{rho_exact}.
\end{corollary}

\subsubsection{Link with the Sharpe formula of \cite{Grebenkov}}
\label{linkwithgrebenkov}

Corollary \ref{cor:grebenkov} and \cite{Grebenkov} describe the same object, namely the Sharpe ratio of the incremental P\&L of the same linear strategy under the same Gaussian process, so they must agree. We show here that they coincide exactly, and that Eq. \eqref{eq:optimal} 
differs from them only through the approximations made by \cite{Grebenkov}, which we list explicitly.

\paragraph{Exact expression}
\cite{Grebenkov} define the risk adjusted P\&L as $\langle\delta\mathcal{P}_\infty\rangle/\sqrt{v_\infty}$, the stationary mean of the incremental P\&L $\delta\mathcal{P}_t=s_tr_t$ divided by the square root of its variance, with (their Eqs. (25) and (28)):
\begin{equation}\label{GS_exact}
	\langle\delta\mathcal{P}_\infty\rangle = \gamma\beta_0^2\,\frac{q}{1-pq},
	\qquad
	v_\infty = \frac{\gamma^2}{1-p^2}\left[1+\frac{2\beta_0^2}{1-pq}+\frac{\beta_0^4\,(1+q^2-2p^2q^2)}{(1-pq)^2}\right].
\end{equation}
The mean is our Eq. \eqref{cov_sr_stationary}. For the variance, Proposition \ref{prop:linear} gives $\mathrm{Var}(H)=\langle s_\infty, s_\infty\rangle\langle r_\infty, r_\infty\rangle+\langle s_\infty, r_\infty\rangle^2$; inserting Eqs. \eqref{cov_stationary}, \eqref{cov_sr_stationary} and \eqref{var_s_stationary} and using
\begin{equation}
	\begin{split}
		&\underbrace{
			(1+\beta_0^2)(1-pq)
			\big(1-pq+\beta_0^2(1+pq)\big)
		}_{\propto\,\langle s_\infty,s_\infty\rangle
			\langle r_\infty,r_\infty\rangle}
		\\[6pt]
		&\quad+
		\underbrace{
			\beta_0^4q^2(1-p^2)
		}_{\propto\,\langle s_\infty,r_\infty\rangle^2}
		\\[6pt]
		&=(1-pq)^2+2\beta_0^2(1-pq)
		\\
		&\quad+\beta_0^4(1+q^2-2p^2q^2).
	\end{split}
	\label{v_decomposition}
\end{equation}
we recover exactly $v_\infty$. Therefore, without any approximation,

\begin{multline}
	SR(H) = \frac{\rho_{xf}}{\sqrt{1+\rho_{xf}^2}}
	= \frac{\langle\delta\mathcal{P}_\infty\rangle}{\sqrt{v_\infty}}
	\\[6pt]
	= \frac{\beta_0^2\,q\,\sqrt{1-p^2}}{
		\sqrt{
			(1-pq)^2+2\beta_0^2(1-pq)
	}}
	\\[6pt]
	\frac{1}{
		\sqrt{
			+\beta_0^4(1+q^2-2p^2q^2)
	}}
\label{exact_identityeq}
\end{multline}

with $\rho_{xf}$ given by Eq. \eqref{rho_exact}: the formula of \cite{Firoozye} for linear rules, evaluated for the process of Eq. \eqref{cov_stationary}, is identical to the exact risk adjusted P\&L of \cite{Grebenkov}. Note that $\gamma$ cancels out, so that the normalization $\gamma^2=1-p^2$ chosen by \cite{Grebenkov} (their Eq. (29)) plays no role in the absence of trading costs.

\paragraph{Approximations of \cite{Grebenkov}}
Eq. \eqref{eq:optimal} is not the exact expression \eqref{exact_identityeq} but Eq. (30) of \cite{Grebenkov}, which they obtain from it ``for $\lambda\ll 1$, $\eta\ll 1$, and $\beta_0^2\ll 1$''. The expansion is made in these three small parameters while the ratio $\beta_0^2/(\lambda+\eta)$ is kept of order one (\cite{Grebenkov} stress that $\beta_0^2/\lambda$ ``is of the order of unity''); this is why the term $2\beta_0^2(\lambda+\eta)$ is retained in the denominator whereas the term in $\beta_0^4$ is not. Three approximations are involved:
\begin{enumerate}
	\item \emph{Continuous-time limit of the mean.} In the numerator, $q=1-\lambda\simeq 1$ and $\sqrt{1-p^2}=\sqrt{\eta(2-\eta)}\simeq\sqrt{2\eta}$, so that $\beta_0^2\,q\sqrt{1-p^2}\simeq\beta_0^2\sqrt{2\eta}$. The relative error is $-\lambda-\eta/4$ to first order.
	\item \emph{Continuous-time limit of the timescales.} In the denominator, $1-pq=\lambda+\eta-\lambda\eta\simeq\lambda+\eta$. The relative error on $1-pq$ is $\lambda\eta/(\lambda+\eta)\leq\min(\lambda,\eta)$.
	\item \emph{Neglect of the term in $\beta_0^4$ in the variance.} To first order, $1+q^2-2p^2q^2\simeq 2\lambda+4\eta$, so that the last term of $v_\infty$ compares to the previous one as
	\begin{equation*}
		\frac{\beta_0^4\,(1+q^2-2p^2q^2)/(1-pq)^2}{2\beta_0^2/(1-pq)} \simeq \beta_0^2\,\frac{\lambda+2\eta}{\lambda+\eta}\in\left[\beta_0^2,\,2\beta_0^2\right],
	\end{equation*}
	and is dropped: $v_\infty\simeq\frac{\gamma^2}{2\eta}\big(1+\frac{2\beta_0^2}{\lambda+\eta}\big)$. By Eq. \eqref{v_decomposition}, the neglected term has two origins: the cross term between the excess variance $\beta_0^2$ of the returns, $\langle r_\infty, r_\infty\rangle=1+\beta_0^2$, and the trend component of the variance of the signal (contribution $\simeq 2\beta_0^4(\lambda+\eta)$), and the squared mean $\langle s_\infty, r_\infty\rangle^2$ which enters the variance of a product of two correlated Gaussian variables, i.e. the factor $1+\rho_{xf}^2\simeq 1$ (contribution $\simeq 2\beta_0^4\eta$).
\end{enumerate}
With these three approximations, Eq. \eqref{exact_identityeq} becomes, without trading costs ($\theta=0$):
\begin{equation}\label{rho_leading}
	\mathbf{S}(\eta)\big|_{\theta=0} = \frac{\beta_0^2\sqrt{2\eta}}{\sqrt{(\lambda+\eta)^2 + 2\beta_0^2(\lambda+\eta)}} ,
\end{equation}
which is Eq. \eqref{eq:optimal} and also the leading order of the correlation $\rho_{xf}$ of Eq. \eqref{rho_exact}. In other words, to this order the theoretical Sharpe ratio of \cite{Grebenkov} is simply the correlation between the EMA signal and the next return.

The trading cost term of Eq. \eqref{eq:optimal} relies on two further approximations of the mean turnover $\langle\mathcal{T}_\infty\rangle$ (Appendix A of \cite{Grebenkov}, linear costs): the continuous-time limit, and the neglect of the contribution of the trend to the turnover, of relative order $\beta_0^2\lambda/(\lambda+\eta)$, so that $\langle\mathcal{T}_\infty\rangle\simeq\frac{2}{\sqrt{\pi}}\theta\sqrt{\eta}$; it is then divided by the same approximate $\sqrt{v_\infty}$.

\paragraph{Size of the approximations}
All the neglected terms are of relative order $\eta$, $\lambda$ or $\beta_0^2$. For the parameters fitted in the main text ($\lambda=1/180$, $\beta_0=0.12$), Eq. \eqref{rho_leading} overestimates the exact Eq. \eqref{exact_identityeq} by $1.2\%$ to $1.9\%$ for $1/\eta$ between $20$ and $1000$ days (e.g. $1.251$ instead of $1.235$ in annualized terms at $1/\eta=72$), of which $0.5\%$ to $0.8\%$ comes from the term in $\beta_0^4$ and the rest from the continuous-time limit. This is negligible compared with the statistical error on the empirical Sharpe ratios, and the dependence on $\eta$ is unaffected in practice.

Finally, both the exact and the approximate expressions are Sharpe ratios of the \emph{incremental} (daily) P\&L, which \cite{Grebenkov} ``conventionally'' multiply by $\sqrt{255}$ to annualize. This convention ignores the auto-correlation of the incremental P\&Ls, which is present in the model and is quantified by the variogram of the P\&L derived in Appendix B of \cite{Grebenkov}.

Two consequences follow. First, comparing Eq. \eqref{monstruoustwo} with Eq. \eqref{SR_linear_grebenkov}, the binary rule and the linear rule have, under the Gaussian assumption and to leading order in $\rho_{xf}$, Sharpe ratios in a fixed ratio independent of $\eta$, $\lambda$ and $\beta_0$:
\begin{equation}\label{ratio_binary_linear}
	\frac{SR_{\text{binary}}}{SR_{\text{linear}}} \simeq \sqrt{\frac{2}{\pi}} \approx 0.80 ,
\end{equation}
so that both rules share the same optimal $\eta$. Second, $\rho_{xf}$, and therefore the dependence of the Sharpe ratio on $\eta$, is determined by the autocovariance function of the returns only (Eq. \eqref{cov_stationary}): any Gaussian process with the same autocovariance, i.e. a white noise plus a single exponentially decaying component, leads to the same curve.

\subsubsection{{\clf Autoregressive coefficients, Kalman filter and optimal $\EMA$}}
\label{kalmanlink}

{\clf The weights $(1-\lambda)^{t-1-k}$ in Eq. \eqref{eq:autoco} apply to the unobserved shocks $\xi_k$ of the trend, not to past returns. They give the autocorrelation of the returns, which decays at the rate $\lambda$ (Eq. \eqref{cov_stationary}), but they are not the autoregressive coefficients of the returns. To obtain the latter, write $r_t=\varepsilon_t+\mu_t$ with $\mu_t=q\,\mu_{t-1}+\beta\,\xi_{t-1}$, so that
\begin{equation}
	r_t - q\,r_{t-1} = \varepsilon_t - q\,\varepsilon_{t-1} + \beta\,\xi_{t-1}.
\end{equation}
The right-hand side is a moving average of order one. It can be written $u_t-\psi\,u_{t-1}$, where $u_t$ is a white noise (the innovation of the returns) and $0<\psi<1$ is obtained by matching the variance and the lag-one covariance:
\begin{equation}\label{psi_equation}
	\psi+\frac{1}{\psi} = \frac{1+q^2+\beta^2}{q}.
\end{equation}
Inverting this relation gives the autoregressive representation of the returns:
\begin{equation}\label{ar_representation}
	r_t = (q-\psi)\sum_{k\geq 1}\psi^{\,k-1}\,r_{t-k} + u_t .
\end{equation}
The autoregressive coefficients therefore decay at the rate $1-\psi$, and the best predictor of $r_t$ given past returns, which is the stationary Kalman filter of this state-space model \citep{Harvey}, is an $\EMA$ of past returns with parameter $\eta^{*}=1-\psi$. For $\lambda\ll 1$ and $\beta_0^2\ll 1$, Eq. \eqref{psi_equation} reads $\eta^{*2}\simeq\lambda^2+2\lambda\beta_0^2$, that is
\begin{equation}
	\eta^{*} \simeq \lambda\sqrt{1+\frac{2\beta_0^2}{\lambda}} = \eta_{\text{opt}} ,
\end{equation}
which is Eq. \eqref{eq:optimal1}. The correspondence is exact before this approximation: the best predictor maximizes the correlation with the next return among all linear functions of past returns, so that the correlation $\rho_{xf}$ of Eq. \eqref{rho_exact} is maximal for $\eta=1-\psi$; and since the Sharpe ratios of the linear rule (Eq. \eqref{SR_linear_grebenkov}) and of the binary rule (Eq. \eqref{monstruoustwo}) are both increasing functions of $\rho_{xf}$, they are maximal for the same $\EMA$. For the parameters fitted in the main text, $1/(1-\psi)=72.8$ days against $72.4$ days for Eq. \eqref{eq:optimal1}. In the one-dimensional Gaussian case without trading costs, the optimal $\EMA$ of \cite{Grebenkov}, the best predictor in the sense of \cite{Zakamuli} and the optimal indicator of the binary rule therefore coincide. They need not coincide once trading costs, a non-zero mean return or non-Gaussian returns are taken into account.}

\subsection{	\cite{Ferreira18}} 
\cite{Ferreira18}, which cited \cite{Grebenkov}, used also a linear trading rule (Eq. \eqref{eq:linearFerreira} ) in their formula (3) from Section 2.2.2, "Risk and Returns".

\begin{equation}
	\langle R \rangle = \frac{1}{T - N + 1} \sum_{t=N}^{T} m_{t-1}(N) \, X_t
	\label{eq:linearFerreira}
\end{equation} 	

In Eq. \eqref{eq:linearFerreira}, $m$ is a simple moving average (SMA) of returns rather than an exponential moving average (EMA). \cite{Ferreira18} derived the Sharpe ratio over two pages without using the results of \cite{Acar}, %
{\clf obtaining a formula that is explicit in the autocorrelations of the returns}. In their formula (Eq. \eqref{eq:SharpeFerreira} ), which appears in Section 2.3, "Limits and Interpretations", the Sharpe ratio depends only on $\rho(t, t-i)$, the correlation coefficient of returns at times $t$ and $t-i$. \cite{Ferreira18} provided the formula closest to \cite{Grebenkov}. The formula is applicable to a more generalized diffusive process as in \cite{Grebenkov}, but it is less straightforward and not entirely consistent, since SMA differs from EMA. \cite{Ferreira18} was published later and includes more parameters in the formula.

\begin{equation}
	SR = 
	\frac{
		\sum_{i=1}^{N} \rho(t, t-i)
	}{
		\sqrt{
			N + \Big( \sum_{i=1}^{N} \rho(t, t-i) \Big)^2 + \sum_{\substack{i,j=1 \\ i \neq j}}^{N} \rho(t-j, t-i)
		}
	}
	\label{eq:SharpeFerreira}
\end{equation}

\subsection{	\cite{Zakamuli22}} 
\cite{Zakamuli22} is also based on a simple moving average (SMA) of returns, Eq. \eqref{eq:momzakamuli} ("2. Time Series Momentum Strategies") but employs a binary strategy in Eq. \eqref{eq:binaryzakamuli}. They derive the Sharpe ratio in Eq. \eqref{eq:sharpezakamuli} using the same approach as \cite{Acar}, which requires the correlation between returns and the forecast ($\rho_n$ is "the correlation coefficient between $r_t$ and the $\MOM_{t-1}(n)$ ). %
{\clf As for \cite{Acar}, this correlation has to be computed for each return process.}

\begin{equation}
	\MOM_{t-1}(n) = \sum_{i=1}^{n} X_{t-i}
	\label{eq:momzakamuli}
\end{equation}

\begin{equation}
	R_t^{LS} =
	\begin{cases}
		r_t, & \text{if } \MOM_{t-1}(n) > 0,\\[2mm]
		2 r_{f,t} - r_t, & \text{otherwise.}
	\end{cases}
	\label{eq:binaryzakamuli}
\end{equation}

\begin{eqnarray}
	E[R_t^{LS}] &=& (2 \, \phi(-d) - 1)\, \mu + 2 \, \big( g + \phi(d) \, r_f \big), \\
	\text{Var}[R_t^{LS}] &=& (\mu^2 + \sigma^2) + 4 r_f \big( g - (\mu - r_f) \phi(d) \big) - \big( E[R_t^{LS}] \big)^2, \\
	g &=& \sigma \, \rho_n \, \phi(d)
	\label{eq:sharpezakamuli}
\end{eqnarray}

\section{Test with a second time scale}
\label{secondtimescale}
{ The good fit of Fig.\ref{fig:empirical_theory} shows that the
	measurements are compatible with a single time scale, but compatibility is not
	a test: a model with two time scales could fit at least as well. Such a model
	does not require a new derivation. As shown in
	Appendix \ref{linkwithgrebenkov}, $\mathbf{S}\left(\eta\right)$ depends on the
	return process only through its autocovariance function, which is linear in the
	contributions of independent trends. For $J$ independent trends with time
	scales $\lambda_{j}$ and strengths $\beta_{0,j}$, the autocovariance of
	Eq. \eqref{cov_stationary} becomes
	$\langle r_{t},r_{t-k}\rangle=\sum_{j}\beta_{0,j}^{2}\left(1-\lambda_{j}\right)^{k}$
	for $k\geq 1$, and the same algebra gives
	\begin{equation}\label{eq:optimalmulti}
		\mathbf{S}\left(\eta\right)=\frac{\sqrt{2\eta}\,\displaystyle\sum_{j=1}^{J}\frac{\beta_{0,j}^{2}}{\lambda_{j}+\eta}-\frac{2}{\sqrt{\pi}}\theta\sqrt{\eta}}
		{\sqrt{1+2\displaystyle\sum_{j=1}^{J}\frac{\beta_{0,j}^{2}}{\lambda_{j}+\eta}}},
	\end{equation}
	which reduces to Eq. \eqref{eq:optimal} for $J=1$.
	
	We fitted Eq. \eqref{eq:optimalmulti} with $J=2$ on the same nine measurements. 
	
	The two time scales collapse onto each other: the fit returns
	$1/\lambda_{1}=180.655$ and $1/\lambda_{2}=180.650$ business days, with
	$\beta_{0,1}=0.1195$ and $\beta_{0,2}=0.0105$. Since
	Eq. \eqref{eq:optimalmulti} depends on two trends sharing the same time scale
	only through the sum $\beta_{0,1}^{2}+\beta_{0,2}^{2}=0.0144$, the fitted model
	is exactly the one-scale model with $\beta_{0}=\sqrt{0.0144}=0.12$, and the
	additional parameters are not identified. The data therefore do not call for a
	second time scale.}

\section{Non Stationarity and Measurement of the decline of $\beta_0$ with time}
\label{decline}

To further investigate the  time variation in the empirical performance of the $\ARP$, we divided the 1990--2023 period into either 2 or 5 equal-length bins, following the same approach as in Fig. 5, but restricting the analysis to the $\ARP$ with linear rules. The results are reported in Fig. \ref{fignonstat}. We observe a gradual decline in the Sharpe ratio over time in both cases: The estimates obtained across successive bins shift progressively downward, resulting in a decline in $\beta_0$ over time, as markets become increasingly efficient.  The estimates also become noisier as the number of bins increases, reflecting the higher sampling variability associated with shorter estimation periods: approximately $\sqrt{2}$ times for 2 bins and $\sqrt{5}$ times for 5 bins, relative to the full-sample estimate.

Fitting Eq. \ref{eq:optimal} to the two-bin data yields $\beta_0=0.159$ and $\lambda=\frac{1}{156}$ for the first period, and $\beta_0=0.078$ and $\lambda=\frac{1}{306}$ for the second period. However, the data are noisier, resulting in a less precise fit.

The standard deviation of the noise of the Sharpe Ratio is 0.24 based on the length of the period for the two-bin case. Accordingly, the standard deviation of the difference between the blue and red lines in panel (a) of Fig. \ref{fignonstat} is $0.24\sqrt{2}\approx 0.34$. The observed difference is 0.95, corresponding to approximately 2.8 standard deviations. This suggests that the blue and red lines are significantly different, with a test statistic of approximately 2.8.

{ The decline of $\beta_0$ over time is both statistically significant and
	economically expected, as markets become more efficient. The apparent change in
	$\lambda$, from $\frac{1}{156}$ to $\frac{1}{306}$, should not be
	over-interpreted: with only half of the sample, the measurement of $\lambda$ is
	too noisy to conclude that it has changed, and we see no economic reason why the
	time scale of trends would have doubled. A change in $\lambda$ is in fact not
	needed to explain Fig.\ref{fignonstat}: according to
	Eq. \eqref{eq:optimal1}, a decline of $\beta_0$ alone, at constant $\lambda$,
	shifts the optimal $\EMA$ towards slower parameters, from about 57 business days
	for $\beta_0=0.159$ to about 101 business days for $\beta_0=0.078$ when
	$\lambda=\frac{1}{180}$. We therefore interpret the data as a decline of
	$\beta_0$ at constant $\lambda$.}

The decline of $\beta_0$ is consistent with the financialization of commodity
futures markets after 2004 \citep{TangXiong2012,ChengXiong2014}: more
speculative capital now faces a hedging demand that has not grown in
proportion. Commercial hedgers, who tend to trade against price changes
\citep{ChengXiongHedgers2014}, are the natural counterparties of trend
followers, and their capacity to absorb risk varies over time
\citep{ChengKirilenkoXiong2015}.

 \begin{figure}[H]
 	\centering
 	\begin{subfigure}[t]{0.48\textwidth}
 		\centering
 		\includegraphics[width=\textwidth]{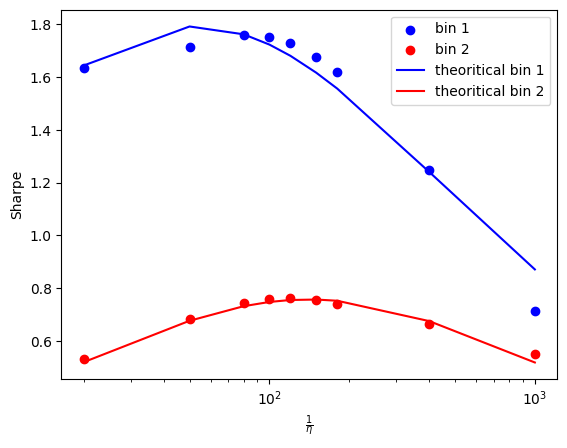}
 		\caption{2 equal-period bins}
 		\label{fig}
 	\end{subfigure}
 	\hfill
 	\begin{subfigure}[t]{0.48\textwidth}
 		\centering
 		\includegraphics[width=\textwidth]{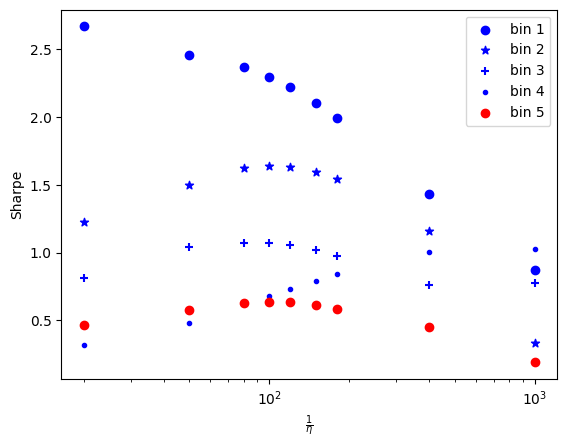}
 		\caption{5 equal-period bins}
 		\label{fig}
 	\end{subfigure}
 	\caption{Empirical Sharpe ratios over the period 1990--2023 for trend-following strategies using the $\ARP$ (Eq. \eqref{eq:arp} ), based on 2 and 5 equal-period bins.}
 	\label{fignonstat}
 \end{figure}

\section{Data}
\label{Data}

We use the most liquid contract of each future(most often the front one), rolled before expiry, so that the
short-term effects across maturities documented by \cite{Rossi2026} do not
enter our measurements.

\begin{table} [H]
	\begin{center}
		\begin{tabular}{|c|} \hline
			Commodities \\  \hline
			\begin{tabular}[t]{l}
				
				Brent Crude (IFEU \$/bbl) \\
				Cocoa (IFUS \$/mt) \\
				Coffee (IFUS \$/lbs) \\
				Corn (CBT \$/bu) \\
				Cotton \#2 (IFUS \$/lbs) \\ 
				Crude Oil WTI (NYM \$/bbl) \\ 
				ECX EUA (IFEU EUR/t) \\
				Feeder Cattle (CME \$/lbs) \\ 
				Gasoil (IFEU \$/mt) \\
				Gold (NYM \$/ozt) \\
				Hard Red Wtr Wheat (CBT \$/bu) \\
				High Grade Copper (NYM \$/lbs) \\
				Iron Ore 62\% Fe, CFR China (TSI) (NYM \$/mt) \\ 
				Lean Hogs (CME \$/lbs) \\
				Live Cattle (CME \$/lbs) \\
				Lumber (CME \$/bft) \\
				Milling Wheat (LIF EUR/t) \\ 
				Natural Gas (NYM \$/mmbtu) \\
				NY Harb RBOB (NYM \$/gal) \\
				NY Harbor ULSD (NYM \$/gal) \\
				Oats (CBT \$/bu) \\
				Orange Juice (IFUS \$/lbs) \\ 
				Palladium (NYM \$/ozt) \\
				Platinum (NYM \$/ozt) \\
				Rough Rice (CBT \$/cwt) \\
				Rubber RSS3 (TKT JPY/kg) \\
				Silver (NYM \$/ozt) \\
				Soybean Meal (CBT \$/t) \\ 
				Soybean Oil (CBT \$/lbs) \\
				Soybeans (CBT \$/bu) \\
				Sugar \#11 (IFUS \$/lbs) \\ 
				Wheat (Chicago) -  Contract \\
				
			\end{tabular}  \\  \hline
		\end{tabular}
	\end{center}
	\caption{
		List of the 32 commodities futures traded in the USA or in Europe. }
	\label{tab:listofinsturmentcommo}
\end{table}

\begin{table}[H]
	\begin{center}
		\begin{turn}{90}
			\begin{tabular}{|c|c|c|} \hline
				Stock indices & Bond indices & FOREX \\  \hline
				\begin{tabular}[t]{l}
					AMSTERDAM EOE Idx\\
					S\&P Midcap 400 Idx e-mini\\
					Russell 2000 Idx e-mini\\
					Cac 40\\
					Dax\\
					Ftse 100\\
					STOXX Europe 600 Index Futures\\
					{\clh Hang Seng}\\
					Mib S\&p-mif\\
					Nikkei 225 Osaka\\
					Topix\\
					Kospi 200\\
					Ibex 35\\
					Mini MSCI Emerging Markets Index Future\\
					Nasdaq E-mini\\
					S\&P 500 e-mini\\
					Dj Euro Stoxx\\
					S\&P Canada 60-ME\\
					SPI 200 Idx\\
					Mini Dow Futures \\
				\end{tabular} &
				\begin{tabular}[t]{l}
					BUND 10Yr\\
					CAD Bond 10Yr\\
					Bobl\\
					Schatz\\
					Long-term Euro-btp\\
					Euro-buxl Futures\\
					LONG Gilt 10Yr\\
					10yr Fr Gov Bond\\
					US T-NOTE 5Yr\\
					JGB 10Yr\\
					US T-NOTE 10Yr\\
					US T-Note 2Yr\\
					Ultra T-Bonds Combined\\
				\end{tabular} &
				\begin{tabular}[t]{l}
					
					AUD/USD Fut.\\
					GBP/USD Fut.\\
					CAD/USD Fut.\\
					EUR/USD Fut.\\
					JPY/USD Fut.\\
					MXN/USD Fut.\\
					NZD/USD Fut.\\
					CHF/USD Fut.\\
				\end{tabular} \\  \hline
			\end{tabular}
		\end{turn}
	\end{center}
	\caption{
		List of the 41 instruments among currencies, equity indices and Bonds. }
	\label{tab:listofinsturment}
\end{table}

\section{Parameters}
\label{Parameters}
\begin{table} [H]
	\begin{center}
	\begin{turn}{90}
	\begin{tabular}{|c|cc|}
		\hline
		portfolio & comment & indicator type  \\
		\hline
		
	 $\ARP\left(20\right)$ &for 20 days as relaxation time for the $\EMA$ when applying & $\EMA\left(\frac{1}{20}\right)$   \\              
 $\ARP\left(50\right)$& when applying &$\EMA\left(\frac{1}{50}\right)$  \\
	 $\ARP\left( 80 \right)$ &when applying& $\EMA\left(\frac{1}{80}\right)$  \\
		 $\ARP\left( 100 \right)$ &when applying &$\EMA\left(\frac{1}{100}\right)$  \\
	 $\ARP\left( 120 \right)$ &when applying &$\EMA\left(\frac{1}{120}\right)$  \\
		 $\ARP\left(150  \right)$& when applying& $\EMA\left(\frac{1}{150}\right)$  \\
		 $\ARP\left( 180  \right)$& when applying& $\EMA\left(\frac{1}{180}\right)$  \\
		 $\ARP\left( 400  \right)$ &when applying&$\EMA\left(\frac{1}{400}\right)$  \\
		 $\ARP\left( 1000  \right)$& when applying& $\EMA\left(\frac{1}{1000}\right)$  \\
	  $\ARP\left( 20, 120,0\times400\right)$ &when applying& $\MACD\left(\frac{1}{20},\frac{1}{120},\frac{1}{400},\omega_1,1,0\right)$  \\
		 $\ARP\left( 0\times20, 120,0\times400\right)$& when applying &$\MACD\left(\frac{1}{20},\frac{1}{120},\frac{1}{400},0,1,0\right)$  \\
	 $\ARP\left( 20, 120,0.2\times400\right)$& when applying& $\MACD\left(\frac{1}{20},\frac{1}{120},\frac{1}{400},\omega_1,1,0.2\right)$  \\
		 $\ARP\left( 20, 120,0.4\times400\right)$ &when applying& $\MACD\left(\frac{1}{20},\frac{1}{120},\frac{1}{400},\omega_1,1,0.4\right)$  \\
	 $\ARP\left( 20, 90,0.3\times400\right)$& when applying& $\MACD\left(\frac{1}{20},\frac{1}{90},\frac{1}{400},\omega_1,1,0.3\right)$  \\
		 $\ARP\left( 20, 80,0.3\times400\right)$& when applying& $\MACD\left(\frac{1}{20},\frac{1}{80},\frac{1}{400},\omega_1,1,0.3\right)$  \\
	 $\ARP\left( 20, 80,0.2\times400\right)$ &when applying& $\MACD\left(\frac{1}{20},\frac{1}{80},\frac{1}{400},\omega_1,1,0.2\right)$  \\
	 $\ARP\left( 20, 80,0.4\times400\right)$ &when applying &$\MACD\left(\frac{1}{20},\frac{1}{80},\frac{1}{400},\omega_1,1,0.4\right)$  \\
		
		\hline
	\end{tabular}
	\end{turn}
	\end{center}
	\caption{Different parameters fo the trend indicator}
	\label{tab:tab0}
\end{table}

\section{Results}
\label{results}

\begin{table} [H]
	\begin{center}
	\begin{turn}{90}
	\begin{tabular}{|c|c|c|c|}
		\hline
		portfolio & indicator type & Gross Sharpe ratio & average holding period (days)) \\
		\hline
		
		$\ARP\left( 20 \right)$ & $\EMA$  &                 1.079535 & 38\\
		$\ARP\left( 50\right)$ & $\EMA$   &                 1.189320 & 60 \\
		$\ARP\left( 80 \right)$ & $\EMA$   &                1.240349 & 74\\
		$\ARP\left( 100 \right)$ & $\EMA$  &                1.244945 &  81\\
		$\ARP\left( 120 \right)$ & $\EMA$   &               1.235455 & 88\\
		$\ARP\left( 150 \right)$ & $\EMA$   &               1.207496 &  96\\
		$\ARP\left( 180 \right)$ & $\EMA$  &                1.172569 & \\
		$\ARP\left( 400 \right)$ & $\EMA$   &               0.955223 & 132\\
		$\ARP\left( 1000 \right)$ & $\EMA$  &               0.633678 & 155\\
		$\ARP\left( 0\times20, 120,0\times400\right)$ & $\MACD$  &   1.235455 & 88\\
		$\ARP\left( 20, 120,0.2\times400\right)$ & $\MACD$&    1.203418 & 97\\
		$\ARP\left( 20, 120,0.4\times400\right)$ & $\MACD$ &   1.176466 & \\
		$\ARP\left( 20, 90,0.3\times400\right)$ & $\MACD$ &    1.214172 & 89\\
		$\ARP\left( 20, 80,0.3\times400\right)$ & $\MACD$&     1.218031 & 85\\
		$\ARP\left( 20, 80,0.2\times400\right)$ & $\MACD$ &    1.228186 & 82\\
		$\ARP\left( 20, 80,0.4\times400\right)$ & $\MACD$ &    1.206864 & 87\\
		
		\hline
	\end{tabular}
		\end{turn}
\end{center}
	\caption{Gross sharpe ratio based on the whole period 1990-2023}
	\label{tab:tab1}
\end{table}

 \begin{figure} [H]
	\centering
	\includegraphics[width=0.8\textwidth]{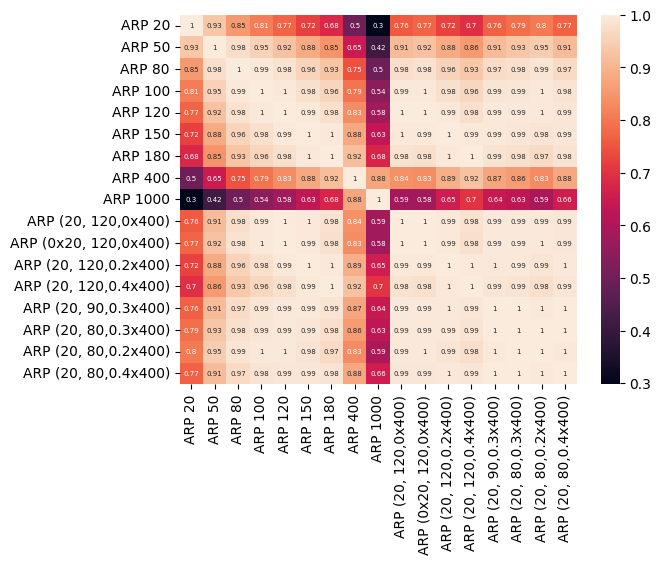}
	\caption{Empirical correlation on the whole period 1990-2023 between ARP models with different trend indicators parameters}
	\label{fig:heat_map}
\end{figure}



\begin{thebibliography}{5}
	
	

	
	
	
	
	
	
	
	
		
	\bibitem[Acar (1998)]{Acar98} E. Acar,  ``Expected returns of directional forecasters'', in Acar and Satchell eds,
	Advanced Trading Rules, Butterworth-Heinemann, Oxford, 51–80, 1998
	
	
	\bibitem[Acar (2003)]{Acar} E. Acar, ``Chapter 8 - modelling directional hedge funds-mean, variance and correlation
	with tracker funds''. In S. Satchell-A. Scowcroft (Eds.), Advances in portfolio construction and implementation (pp. 193–214). Oxford: Butterworth-Heinemann, 2003
	
\bibitem[Ayed et al (2016)]{Ayed} A. Bel Hadj Ayed, G. Loeper and F. Abergel ``Forecasting trends with asset prices'', \textit{Quantitative Finance}, 17(3), 369–382. https://doi.org/10.1080/14697688.2016.1206959, 2016





	\bibitem[Baltas (2015)]{Baltas}	N. Baltas,   ``Trend-following, Risk-parity and the influence of Correlations'',	In E. Jurczenko (ed.),	Risk-based and Factor Investing, ISTE Press \& Elsevier {\bf 3}, 65–96, 2015
	\bibitem[Benichou et al(2017)] {Benichou16}		R. Benichou, Y. Lemp\'eri\`ere, E. S\'eri\'e, J. Kockelkoren, P. Seager, J.-P. Bouchaud, and M. Potters,	``Agnostic Risk Parity: Taming Known and Unknown-Unknowns'',	\textit{Journal of Investments Strategies}, {\bf 6} (3), 1-12, 2017
	
	\bibitem[Benveniste et al(2024)] {Benveniste}		J. Benveniste, P. N. Kolm and G. Ritter,	``Untangling Universality and Dispelling Myths in Mean-Variance Optimization'', \textit{The Journal of Portfolio Management},  Special Issue Dedicated to Harry Markowitz, 50  ( 8) 90 - 116, 2024 
		\bibitem[Black and Litterman (1992)]{BlackLitterman92}
	F. Black and R. Litterman, ``Global Portfolio Optimization'', \textit{Financial Analysts Journal}, 48(5), 28–43, 1992.
	\bibitem[Boons and Prado(2019)]{BoonsPrado2019}
	M. Boons and M. P. Prado, ``Basis-Momentum'', The Journal of Finance, 74,
	239--279, 2019
	\bibitem[Bouchaud et al (2017)]{Bouchaud2017}	J.-P. Bouchaud, S. Ciliberti, Y. Lemperiere, A. Majewski, P. Seager, S. K. Ronia, ``Black was right: Price is within a factor 2 of Value'',  https://doi.org/10.48550/arXiv.1711.04717, 2017
	\bibitem[Brandt and Santa Clara (2006)]{Brandt} M W. Brandt and P. Santa Clara, ``Dynamic Portfolio Selection by Augmenting the Asset Space'', {\bf 61}, 2187-2217,   https://doi.org/10.1111/j.1540-6261.2006.01055.x, \textit{Journal of Finance}, 2006
	\bibitem[Bun et al(2016)]{Bun16}	J. Bun,  J.-P. Bouchaud, and M. Potters, ``Cleaning correlation matrices'',	\textit{Risk publication}, 2016
\bibitem[Buncic(2025)]{Buncic}	D. Buncic, ``Simplified: A closer Look at the Virtue of {\clh Complexity} in Return Prediction'', working paper,	SSRN, 2025
\bibitem[Cheng and Xiong(2014a)]{ChengXiong2014}
I.-H. Cheng and W. Xiong, ``Financialization of Commodity Markets'', Annual
Review of Financial Economics, 6, 419--441, 2014

\bibitem[Cheng and Xiong(2014b)]{ChengXiongHedgers2014}
I.-H. Cheng and W. Xiong, ``Why Do Hedgers Trade So Much?'', The Journal of
Legal Studies, 43(S2), S183--S207, 2014

\bibitem[Cheng et al(2015)]{ChengKirilenkoXiong2015}
I.-H. Cheng, A. Kirilenko and W. Xiong, ``Convective Risk Flows in Commodity
Futures Markets'', Review of Finance, 19(5), 1733--1781, 2015

	
	\bibitem[Dai et al(2010)]{Dai10}  M. Dai, Q. Zhang, and Q. J. Zhu , ``Trend Following Trading under a Regime Switching
	Model'', \textit{SIAM Journal on Financial Mathematics},{\bf 1}(1), 780–810, 2010
	\bibitem[Dai et al(2016)]{Dai16}  M. Dai, Z. Yang, Q. Zhang and Q. J. Zhu,   ``Optimal Trend Following Trading
	Rules'', \textit{Mathematics of Operations Research}, {\bf 41}(2), 626–642, 2016
	

	\bibitem[Elder (2025)] {Elder} B. Elder, ``Are bigger AI models better stock picker? Maybe probably not'', \textit{financial times}, 2025
\bibitem[El Ghaoui et al. (2003)]{ElGhaoui03}
L. El Ghaoui, M. Oks, and F. Oustry, ``Worst-Case Value-at-Risk and Robust Portfolio Optimization: A Conic Programming Approach'',  \textit{Operations Research}, 51(4), 543–556, 2003.


	\bibitem[Ferreira et al.(2018)]{Ferreira18} F. F. Ferreira, A. C. Silva and J. Y. Yen, ``Detailed study of a moving average trading rule'', \textit{Quantitative Finance}, 18(9), 1599–1617, 2018
	
		\bibitem[Ferson and Siegel (2001)]{Ferson}  W. E. Ferson,  A.F Siegel, ``The Efficient Use of Conditioning Information in Portfolios'', \textit{Journal of Finance}, 56 (3), 19–38, 2001
	
	\bibitem[Firoozye and Koshiyama (2020)]{Firoozye} N. Firoozye and A. Koshiyama, ``Optimal Dynamic Strategies on Gaussian Returns'', \textit{Journal of Investment Strategies}, {\bf 9}, 23-53,  10.21314/JOIS.2020.118, 2020
	\bibitem[Firoozye et al (2023)]{Firoozye23} N. Firoozye, V. Tan and S. Zohren, ``Canonical Portfolios: Optimal Asset and Signal Combination'', \textit{Journal of Banking \& Finance}, {\bf154}, jbankfin.2023.106952, 2023
	\bibitem[Fruhwirth-Schnatter (2006)] {Fruhwirth-Schnatter} S. Fruhwirth-Schnatter, ``Finite Mixture and Markov Switching Models'', Springer, New-York, 2006
	\bibitem[Giner and Zakamulin (2023)] {Giner} J. Giner and V. Zakamulin, ``A Regime-Switching Model of Stock Returns with    Momentum and Mean Reversion'', \textit{Economic Modelling}, {\bf 122}, 106237, 2023
	\bibitem[Gmür et al (2025)]{quantica}   B. Gmür, N. Mirjolet and L. Spiga,   ``The Speed Factor-  How Trend Model Speed Has Driven Performance Dispersion '', Quantica Quarterly Insights, November 2025
	\bibitem[Goldfarb and Iyengar (2003)]{Goldfarb03}
D. Goldfarb and G. Iyengar, ``Robust Portfolio Selection Problems'', \textit{Mathematics of Operations Research}, 28(1), 1–38, 2003.

	\bibitem[Grebenkov and Serror (2014)]{Grebenkov}		D. Grebenkov and J. Serror, 	``Following a trend with an exponential moving average: Analytical results for a Gaussian model'', 	\textit{Physica A: Statistical Mechanics and its Applications}, {\bf 394}, 288-303, 2014
	\bibitem[Grebenkov and Serror (2015)]{Grebenkov15}		D. Grebenkov and J. Serror, ``Optimal Allocation of Trend Following Strategies'', 	\textit{Physica A: Statistical Mechanics and its Applications}, {\bf 433}, 107-125, 2015.
	
	\bibitem[Harvey(1990)]{Harvey} A C.	Harvey  ``Forecasting, Structural Time Series Models and the Kalman Filter'', Cambridge University Press, 1990.
	\bibitem[Harvey et al(2016)]{HarveyLiuZhu2016}
	C. R. Harvey, Y. Liu and H. Zhu, ``\ldots and the Cross-Section of Expected
	Returns'', The Review of Financial Studies, 29(1), 5--68,
	https://doi.org/10.1093/rfs/hhv059, 2016
	\bibitem[Hong and Stein(1999)]{HongStein1999}
	H. Hong and J. C. Stein, ``A unified theory of underreaction, momentum trading,
	and overreaction in asset markets'', The Journal of Finance, 54, 2143--2184,
	1999
	\bibitem[Hurst et al(2013)]{Hurst13}		B. Hurst, Y. H. Ooi and L. H. {\clh Pedersen}, ``Demystifying managed futures'', \textit{The Journal of portfolio management}, {\bf 11}, 42-58, 2013
	\bibitem[Hurst et al(2017)]{Hurst}		B. Hurst, Y. H. Ooi and L. H. {\clh Pedersen}, ``A Century of Evidence on Trend-Following Investing'', \textit{The Journal of portfolio management}, {\bf 44}, 2017
			\bibitem[Jagannathan and Ma (2003)]{JagannathanMa03}
		R. Jagannathan and T. Ma, ``Risk Reduction in Large Portfolios: Why Imposing the Wrong Constraints Helps'',  \textit{Journal of Finance}, 58(4), 1651–1683, 2003.
		
		\bibitem[Jorion (1986)]{Jorion86}
		P. Jorion, ``Bayes-Stein Estimation for Portfolio Analysis'',  \textit{Journal of Financial and Quantitative Analysis}, 21(3), 279–292, 1986.
			\bibitem[Kan and Zhou (2007)]{KanZhou07}
		R. Kan and G. Zhou, ``Optimal Portfolio Choice with Parameter Uncertainty'',  \textit{Journal of Financial and Quantitative Analysis}, 42(3), 621–656, 2007.

\bibitem[Kang et al(2020)]{Kang2020}
W. Kang, K. G. Rouwenhorst and K. Tang, ``A Tale of Two Premiums: The Role of
Hedgers and Speculators in Commodity Futures Markets'', The Journal of Finance,
75, 377--417, 2020

		\bibitem[Kelly et al(2022)]{Kelly22}		B. Kelly, S. Malamud and L. Pedersen, 	``Principal Portfolios'', \textit{The Journal of finance}, {\bf 78},  347-387,https://doi.org/10.1111/jofi.13199, 2022
		\bibitem[Kelly et al(2023)]{Kelly}		B. Kelly, S. Malamud and K. Zhou, 	``The Virtue of Complexity in Return Prediction'', \textit{The Journal of finance}, {\bf 79},  459-503, 2023
	
	\bibitem[Kim and Omberg(1996)]{Kim}		T. S. Kim and E. Omberg, 	``Dynamic Nonmyopic Portfolio Behavior'', \textit{The Review of Financial Studies},, Volume 9, Issue 1, Pages 141–161, https://doi.org/10.1093/rfs/9.1.141, 1996
	\bibitem[Koijen et al(2018)]{Koijen2018}
	R. S. Koijen, T. J. Moskowitz, L. H. Pedersen and E. B. Vrugt, ``Carry'',
	Journal of Financial Economics, 127, 197--225, 2018
	\bibitem[Koshiyama and Firoozye(2019)]{Koshiyama}		A. Koshiyama and N. Firoozye, 	``Avoiding Backtesting Overfitting by Covariance-Penalties: An Empirical Investigation of the Ordinary and Total Least Squares Cases'', \textit{The Journal of Financial Data Science}, {\bf 1},  63-83, 2019
	\bibitem[Kurth et al (2026)]{Kurth}
J. G. Kurth, Z. Eisler, A. Rej and J.-P. Bouchaud, ``Is Trend Still Your Friend?: A Microstructural Account of the Demise of Short-Term Trend-Following'', https://doi.org/10.48550/arXiv.2607.01550, 2026.
\bibitem[Kurth et al (2026b)]{Kurth2} J. G. Kurth, {\clh A. A. Majewski}, J.-P. Bouchaud, ``Revisiting the excess volatility puzzle through the lens of the Chiarella model'', \textit{PLoS One}, 21(1):e0340409. https://doi.org/10.1371/journal.pone.0340409, 2026
	\bibitem[Lakner(1998)]{Lakner}	P. Lakner, 	``Optimal trading strategy for an investor: the case of partial information'',	\textit{Stochastic Processes and their Applications},  Volume 76, Issue 1,Pages 77-97,
	ISSN 0304-4149,	https://doi.org/10.1016/S0304-4149(98)00032-5, 1998
	


	\bibitem[Lemperiere et al(2014)]{Lamperiere} Y. Lempérière, C. Deremble, P. Seager, M. Potters,  and J.-P. Bouchaud, ``Two centuries of trend following'', \textit{Journal of Investment Strategies}, {\bf 3}, 41-61, 2014
		\bibitem[Ledoit and Wolf (2003)]{LedoitWolf03}
	O. Ledoit and M. Wolf, ``Improved Estimation of the Covariance Matrix of Stock Returns with an Application to Portfolio Selection'', \textit{Journal of Empirical Finance}, 10(5), 603–621, 2003.
	
	\bibitem[Ledoit and Wolf (2004)]{LedoitWolf04}
	O. Ledoit and M. Wolf, ``A Well-Conditioned Estimator for Large-Dimensional Covariance Matrices'', \textit{Journal of Multivariate Analysis}, 88(2), 365–411, 2004.
	
	\bibitem[McLean and Pontiff (2016)]{McLeanPontiff2016}	R. D. McLean and J. Pontiff, ``Does Academic Research Destroy Stock Return
	Predictability?'', The Journal of Finance, 71(1), 5--32,
	https://doi.org/10.1111/jofi.12365, 2016
	
	\bibitem[Majewski et al (2020)]{Majewski}	A. A. Majewski, S. Ciliberti, J.-P. Bouchaud,``Co-existence of trend and value in financial markets: Estimating an extended Chiarella model'', \textit{Journal of Economic Dynamics and Control},	Volume 112,103791, 	ISSN 0165-1889,	https://doi.org/10.1016/j.jedc.2019.103791.2020
	
	
	
		\bibitem[Markowitz (1952)]{Markowitz52}
	H. Markowitz, ``Portfolio Selection'', \textit{Journal of Finance}, 7(1), 77–91, 1952.
	
	\bibitem[Miffre and Rallis(2007)]{MiffreRallis2007}
	J. Miffre and G. Rallis, ``Momentum strategies in commodity futures markets'',
	Journal of Banking \& Finance, 31, 1863--1886, 2007
		\bibitem[Moskowitz et al(2012)]{Moskowitz} T. J. Moskowitz, Y. H. Ooi and L. H. Pedersen, ``Time series momentum'', \textit{Journal of Financial Economics}, {\bf 104}, Issue 2, 2012
		
		\bibitem[Moskowitz et al(2025)]{Moskowitz2} T. J. Moskowitz, R. Sabbatucci, A. Tamoni, B. Uhl, ``Nonlinear Time Series Momentum'', Available at SSRN: https://ssrn.com/abstract=5933974 or http://dx.doi.org/10.2139/ssrn.5933974, 2025
		\bibitem[Murphy(1999)]{Murphy} J. J. Murphy, ``Technical Analysis of Financial Markets'', NY Institute of Finance, Penguin Group, 1999
		\bibitem[Nguyen et al(2014a)]{Nguyen14a} D. Nguyen, J. Tie and  Q. Zhang, ``An Optimal Trading Rule Under a Switchable
	Mean-Reversion Model'', \textit{Journal of Optimization Theory and Applications}, {\bf 161}, 145-163, 2014
	
	\bibitem[Nguyen et al(2014b)]{Nguyen14b}   D. Nguyen, G. Yin and Q. Zhang,   ``A Stochastic Approximation Approach for Trend-Following Trading'', In Mamon, R. S. and Elliott, R. J. (Eds.), Hidden Markov
	Models in Finance: Further Developments and Applications, {\bf II},  167–184.
	Springer US, Boston, MA, 2014
	
	\bibitem[Rieder(2005)]{Rieder}   U. Rieder and N. Bauerle,  ``Portfolio Optimization with Unobservable Markov-Modulated Drift Process'', \textit{Journal of Applied Probability}, 42(2), 362–378. http://www.jstor.org/stable/30040798, 2005
	
		
	\bibitem[Rodriguez Dominguez et al(2025)]{RODRIGUEZDOMINGUEZ2025128633}   A. Rodriguez Dominguez, M. Shahzad and X. Hong,   ``Multi-hypothesis prediction for portfolio optimization: A structured ensemble learning approach to risk diversification'', \textit{Expert Systems with Applications}, 292 (128633), 2025
	\bibitem[De Roon et al(2000)]{DeRoon2000}
	F. A. De Roon, T. E. Nijman and C. Veld, ``Hedging pressure effects in futures
	markets'', The Journal of Finance, 55, 1437--1456, 2000
	\bibitem[Rossi et al(2026)]{Rossi2026}
	A. G. Rossi, Y. Zhang and Y. Zhu, ``Short-Term Basis Reversal'', working paper,
	Available at SSRN: https://ssrn.com/abstract=5250499, 2026
	
	\bibitem[Safari and Schmidhuber(2025)]{Safari} S. A. Safari and C. Schmidhuber, ``Trends and Reversion in Financial Markets on Time Scales from Minutes to Decades'', working paper,	arXiv:2501.16772 , 2025
	
	\bibitem[Schmidhuber(2021)]{Schmidhuber} C. Schmidhuber, ``Trends, reversion, and critical phenomena in financial markets.", \textit{Physica A: Statistical Mechanics and its Applications}, {\bf 566},  125642, 2021
	\bibitem[Schmidhuber(2022)]{Schmidhuber22} C. Schmidhuber, ``Financial markets and the phase transition between water and steam'', Physica A: Statistical Mechanics and its Applications {\bf 592}, 126873, 2022
	\bibitem[Segonne(2025)]{Segonne} F. Segonne, ``Basis Immunity: Isotropy as a Regularizer for Uncertainty'', https://arxiv.org/pdf/2511.13334, 2025
		\bibitem[Sepp and Lucic(2025)]{Sepp} A. Sepp and V. Lucic, ``The Science and Practice of Trend-following Systems'' , working paper, SSRN, 2025
		\bibitem[Tang and Xiong(2012)]{TangXiong2012}
		K. Tang and W. Xiong, ``Index Investment and the Financialization of
		Commodities'', Financial Analysts Journal, 68, 54--74, 2012
	\bibitem[Tie and Zhang (2016)]{Tie16} J. Tie  and Q. Zhang,   ``An Optimal Mean-Reversion Trading Rule Under a Markov 	Chain Model'', \textit{Mathematical Control \& Related Fields}, {\bf 6} (3), 467–488, 2016
	
    \bibitem[Timmermann (2000)] {Timmermann} A. Timmermann, ``Moments of Markov Switching Models'', \textit{Journal of Econometrics}, {\bf 96}(1), 75–111, 2000
    	\bibitem[Tzotchev(2018)]{Tzotechev} D. Tzotchev, ``Designing Robust Trend-following System: Behind the Scenes of Trend-following'', working papers, SSRN, 2018
    \bibitem[Valeyre(2024)] {Valeyre}		S. Valeyre,  ``Optimal trend-following portfolios'',     \textit{Journal of Investments Strategies}, {\bf 12}, 1-21,  10.21314/JOIS.2023.008, 2024
    \bibitem[Wilder(1978)] {Wilder}		J. W. Wilder Jr., 	``New Concepts in Technical {\clh Trading Systems, Trend Research}'', Greensboro, North Carolina, 1978
    	\bibitem[Zakamulin and Giner (2024)]{Zakamuli}	V. Zakamulin and J. Giner, ``Optimal Trend Following Rules in Two-State Regime-Switching Models'', \textit{Journal of Asset Management}, Springer, 2024
    	\bibitem[Zakamulin and Giner (2020)]{Zakamuli20}	V. Zakamulin and J. Giner, ``Trend Following with momentum versus moving averages: a tale of differences'', \textit{Quantitative Finance},  20:6, 985-1007, 2020
    	\bibitem[Zakamulin and Giner (2022)]{Zakamuli22}    	V. Zakamulin and J. Giner,  ``Time series momentum in the US stock market: Empirical evidence and theoretical analysis'', \textit{Int. Rev. Financ. Anal.}, 82, 102173, 2022
\bibitem[Zhao et al (2026)]{Zhao}    S. Zhao, Y. Ding, J. Yu and W. Kang, ``Momentum and Reversal on the Short-Term Horizon: Evidence from Commodity Markets'', Available at SSRN: https://ssrn.com/abstract=6425598 or http://dx.doi.org/10.2139/ssrn.6425598, 2026

\end{thebibliography}
\end{document}